\documentclass[journal]{IEEEtran}
\usepackage{amsmath,amssymb,amsfonts}
\usepackage{amsthm}
\usepackage{graphicx}
\usepackage{cite}
\usepackage{algorithmic}
\usepackage{algorithm}
\usepackage{booktabs}
\usepackage{pgfplots}
\usepackage{iftex}
\usepackage{hyperref}

\newtheorem{theorem}{Theorem}

\graphicspath{{sim/}{./}}

\title{Exact Closed-Form Feedforward Inversion for Dual-Bridge Series Resonant DC/DC Converter via State-Plane Analysis}

\author{Alex Borisevich, \href{mailto:akpc806b@gmail.com}{akpc806b@gmail.com}
}

\begin{document}

\maketitle

\begin{abstract}
This paper derives exact closed-form feedforward inversion maps for the dual-bridge series resonant converter (DB SRC) using state-plane trajectory analysis. The converter employs four modulation variables: primary duty cycle $d$, secondary shorting time $s$, phase shift $\beta$, and switching frequency $\omega$. While the established first harmonic approximation (FHA) provides frequency-independent inversion, the exact state-plane approach yields frequency-dependent inversion model that is proven algebraically identical to FHA at resonance frequency. For practical above-resonance operation, the exact inversions eliminate the 5--72\% commutation angle errors inherent in the FHA-based feedforward. The resulting controller architecture mirrors the parallel nonlinear compensation structure of the FHA-based design, with feedforward maps now operating on resonant-time quantities that naturally couple commutation and frequency control. All results are expressed in closed form suitable for real-time implementation.
\end{abstract}


% ==============================================================
\section{Introduction}
% ==============================================================

The dual-bridge (DB) series resonant converter (SRC) offers significant advantages for high-power isolated DC/DC applications including electric vehicle chargers, battery energy storage, and DC power distribution \cite{Li2010, Krismer2005, Wang2017}. The topology features controlled full-bridge circuits on both primary and secondary sides connected through a series LC resonant tank and a transformer, enabling bidirectional power flow, voltage boost capability, and near-sinusoidal transformer current.

The modulation scheme employs four independent control variables: primary bridge duty cycle $d$, secondary bridge shorting time $s$, phase shift $\beta$ between bridges, and switching frequency $\omega$ \cite{Zhao2014}. This provides sufficient degrees of freedom to simultaneously regulate output power and optimize efficiency through soft-switching operation.

Existing analytical models for the DB SRC are predominantly based on the first harmonic approximation (FHA), which assumes purely sinusoidal tank current \cite{Borisevich2019}. While computationally simple and adequate for many operating conditions, the FHA inherently neglects higher harmonics and cannot provide exact values of commutation currents needed for ZVS analysis.

State-plane trajectory analysis, pioneered by Oruganti and Lee \cite{Oruganti1985, Oruganti1988}, provides an exact geometric method for analyzing resonant converters. The key insight is that in normalized coordinates, the LC tank trajectory under constant applied voltage is a circular arc, and the steady-state is found by stitching arcs together with periodicity constraints. This approach has been successfully applied to SRC converters with passive rectifiers \cite{Rossetto1996, Mohammadi2015a, Mohammadi2015b} and to LLC resonant converters \cite{Hu2014}.

However, the state-plane method has not been systematically applied to the DB SRC with all four modulation degrees of freedom ($d$, $s$, $\beta$, $\omega$), where both bridges are actively controlled. The secondary-side modulation introduces additional switching intervals and complicates the trajectory structure compared to passive rectifier topologies.

In this paper, we develop the complete exact state-plane analysis for the DB SRC under generalized modulation and establish its precise relationship to the FHA model from \cite{Borisevich2019}. The main contributions are:

\begin{enumerate}
\item Closed-form expressions for the exact steady-state trajectory, output transconductance, and commutation angles under arbitrary $(d, s, \beta, \omega)$ modulation.
\item A rigorous proof that the FHA is recovered as the fundamental-frequency term of the exact solution, with explicit error bounds.
\item Quantitative characterization of FHA error across the full operating range, identifying regimes where the exact model provides significant improvement.
\end{enumerate}

% ==============================================================
\section{Converter Topology and Modulation}
% ==============================================================

\subsection{Circuit Description}

The DB SRC circuit consists of two full H-bridge inverters connected through a series LC resonant tank and a transformer with turns ratio $n$, as shown in Fig.~\ref{fig:circuit}. The input bridge (switches $S_1$--$S_4$) is connected to DC voltage source $V_{in}$. The output bridge (switches $S_5$--$S_8$) is connected to the secondary winding and delivers current to the output at voltage $V_{out}$.

\begin{figure}[t]
\centering
% no generating script found in sim; static/imported artwork
\includegraphics[width=0.95\columnwidth]{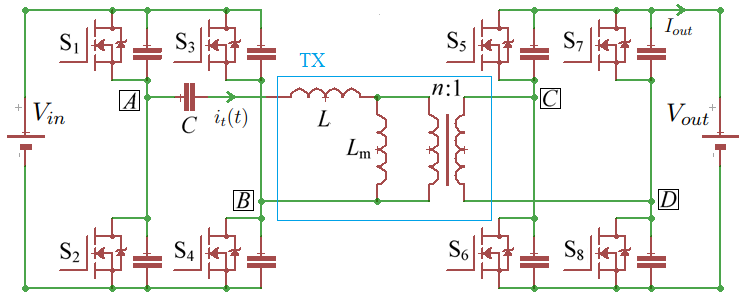}
\caption{Dual-bridge series resonant DC/DC converter circuit.}
\label{fig:circuit}
\end{figure}

The series combination of transformer leakage inductance $L$ and resonant capacitor $C$ forms the LC tank with resonant frequency $\omega_0 = 1/\sqrt{LC}$ and characteristic impedance $Z_0 = \sqrt{L/C}$. The tank current $i_L(t)$ flows through both bridges and the transformer.

\subsection{Switching Waveforms and Control Variables}

The converter is modulated by four independent parameters defined in switching-frequency angular coordinates ($t' = \omega t$):

\begin{itemize}
\item $d \in [0, \pi]$: Primary bridge on-time (duty cycle). Full square-wave corresponds to $d = \pi$.
\item $s \in [0, \pi]$: Secondary bridge short-time. The secondary is shorted at the beginning of each half-cycle for duration $s$.
\item $\beta \in [0, \pi]$: Phase shift between primary and secondary bridge switching cycles.
\item $\omega > \omega_0$: Switching angular frequency (above resonance for ZVS).
\end{itemize}

The voltage gain is defined as:
\begin{equation}\label{eq:gain}
G = \frac{n V_{out}}{V_{in}}
\end{equation}

The primary bridge produces a half-wave-symmetric voltage:
\begin{equation}\label{eq:v_in}
\frac{v_{in}(t')}{V_{in}} = \begin{cases} +1 & 0 < t' < d \\ 0 & d < t' < \pi \\ -1 & \pi < t' < \pi + d \\ 0 & \pi + d < t' < 2\pi \end{cases}
\end{equation}

The secondary bridge, referenced to its own cycle starting at phase $\beta$, produces (with $\phi = t' - \beta \bmod 2\pi$):
\begin{equation}\label{eq:v_out}
\frac{n \cdot v_{out}(t')}{V_{in}} = \begin{cases} 0 & 0 < \phi < s \\ +G & s < \phi < \pi \\ 0 & \pi < \phi < \pi + s \\ -G & \pi + s < \phi < 2\pi \end{cases}
\end{equation}

The net voltage applied to the LC tank is $v(t) = v_{in}(t) - n \cdot v_{out}(t)$.

\subsection{Timing Quantities $\sigma$ and $\delta$}

Two measurable timing quantities characterize the relationship between the tank current zero crossing and bridge switching events (Fig.~\ref{fig:waveforms}):

\begin{itemize}
\item $\sigma$: Angular time from the rising edge of $v_{in}$ to the positive-going zero crossing of $i_L$.
\item $\delta$: Angular time from the zero crossing of $i_L$ to the rising edge of $v_{out}$.
\end{itemize}

These satisfy $\sigma + \delta = \beta$. The soft-switching conditions for ZVS require $\sigma \ge 0$ (ensures negative current at primary turn-on) and $\delta \ge 0$ (ensures positive current at secondary turn-on).

\begin{figure}[t]
\centering
% no generating script found in sim; static/imported artwork
\includegraphics[width=0.85\columnwidth]{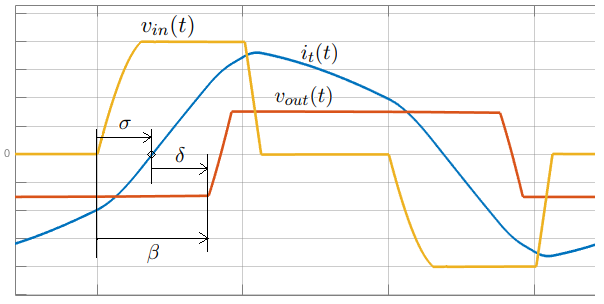}
\caption{Timing of bridge voltages relative to tank current zero crossing, defining $\sigma$ and $\delta$.}
\label{fig:waveforms}
\end{figure}

\subsection{First Harmonic Approximation Model}

The FHA model from \cite{Borisevich2019} approximates the tank current as purely sinusoidal and yields the output transconductance:

\begin{equation}\label{eq:FHA_model}
W_{FHA} = \frac{I_{out}}{V_{in}} = \frac{n}{2\pi^2} \frac{\sqrt{A^2 + B^2}}{Z(\omega)} \left(\cos(s+\delta) + \cos\delta\right)
\end{equation}

where the Fourier coefficients of the applied voltage are:
\begin{equation}\label{eq:FHA_AB}
\begin{gathered}
A = 4\sin d + 4G\sin(\beta+s) + 4G\sin\beta \\
B = 4 - 4G\cos(\beta+s) - 4G\cos\beta - 4\cos d
\end{gathered}
\end{equation}

the tank impedance is $Z(\omega) = \omega L - 1/(\omega C)$, and the timing quantities are:
\begin{equation}\label{eq:FHA_sigma}
\sigma = \text{atan2}(B, A), \quad \delta = \beta - \sigma
\end{equation}

% ==============================================================
\section{Exact State-Plane Analysis}
% ==============================================================

\subsection{Complex State Variable}

Introducing normalized coordinates $x = v_C/V_{in}$, $y = Z_0 i_L/V_{in}$, $\theta = \omega_0 t$, the LC dynamics become $dx/d\theta = y$, $dy/d\theta = c - x$ where $c$ is the normalized applied voltage (constant per interval).

Define the complex state $z = x + jy$. For constant $c$, the exact solution is:
\begin{equation}\label{eq:complex_arc}
z(\theta) = c + (z_0 - c)\,e^{-j\theta}
\end{equation}

This traces a clockwise circular arc centered at $c$ (on the real axis) with radius $|z_0 - c|$.

\subsection{Interval Decomposition}

The positive half-cycle decomposes into $N$ intervals with centers $c_k$ and switching-angle durations $\alpha_k$ (satisfying $\sum_k \alpha_k = \pi$). In resonant-time, the durations are $\Delta\theta_k = \alpha_k/F_n$ and the half-period is $\Theta = \pi/F_n$.

For the DB SRC, two principal orderings exist:

\textbf{Case 1} ($d \ge s+\beta$):
$c_k = \{1{+}G,\; 1,\; 1{-}G,\; {-}G\}$, \; $\alpha_k = \{\beta,\; s,\; d{-}s{-}\beta,\; \pi{-}d\}$

\textbf{Case 2} ($d < s+\beta$):
$c_k = \{1{+}G,\; 1,\; 0,\; {-}G\}$, \; $\alpha_k = \{\beta,\; d{-}\beta,\; s{-}d{+}\beta,\; \pi{-}s{-}\beta\}$

\subsection{Composition and Periodicity}

Each interval applies the map $z \mapsto c_k + (z - c_k)e^{-j\Delta\theta_k}$. Composing all $N$ intervals:
\begin{equation}\label{eq:composition}
z_N = e^{-j\Theta}\,z_0 + \mathcal{C}
\end{equation}

where the complex offset $\mathcal{C}$ satisfies the recurrence:
\begin{equation}\label{eq:recurrence}
\mathcal{C}_0 = 0, \quad \mathcal{C}_k = (\mathcal{C}_{k-1} + c_k)\,e^{-j\Delta\theta_k} - c_k, \quad \mathcal{C} = \mathcal{C}_N
\end{equation}

Half-wave symmetry requires $z_N = -z_0$, giving (see Appendix~A for derivation):
\begin{equation}\label{eq:z0}
z_0 = \frac{\mathcal{C}}{1 + e^{-j\Theta}} = \frac{\mathcal{C}}{2\cos(\Theta/2)}\,e^{j\Theta/2}
\end{equation}

This is the exact closed-form initial state. The explicit expansion of $\mathcal{C}$ is:
\begin{equation}\label{eq:C_expanded}
\begin{split}
\mathcal{C} &= \sum_{k=1}^N c_k(e^{-j\Delta\theta_k}-1)\prod_{m=k+1}^N e^{-j\Delta\theta_m} \\ &= -2j\sum_{k=1}^N c_k\sin\frac{\Delta\theta_k}{2}\,e^{-j(\Delta\theta_k/2+\Phi_k)}
\end{split}
\end{equation}

where $\Phi_k = \sum_{m=k+1}^N \Delta\theta_m$.

\subsection{Output Transconductance}

Using $\int y\,d\theta = \Delta x = \mathrm{Re}(\Delta z)$ per interval and the rectification sign $r_k \in \{+1, -1, 0\}$:
\begin{equation}\label{eq:W_half}
W = \frac{n}{Z_0\Theta}\sum_{k=1}^N r_k\,\mathrm{Re}(z_k - z_{k-1})
\end{equation}

Since $z_k - z_{k-1} = (z_{k-1}-c_k)(e^{-j\Delta\theta_k}-1)$ and $z_{k-1} = e^{-j\Theta_{k-1}}z_0 + \mathcal{C}_{k-1}$, the transconductance is linear in $z_0$:
\begin{equation}\label{eq:W_linear}
W = \frac{n}{Z_0\Theta}\,\mathrm{Re}(z_0\cdot\mathcal{A} + \mathcal{B})
\end{equation}

with coefficients:
\begin{equation}\label{eq:AB_def}
\begin{split}
\mathcal{A} = \sum_{k=1}^N r_k\,e^{-j\Theta_{k-1}}(e^{-j\Delta\theta_k}-1), \\ \mathcal{B} = \sum_{k=1}^N r_k\,(-\mathcal{C}_{k-1}-c_k)(e^{-j\Delta\theta_k}-1)
\end{split}
\end{equation}

Substituting $z_0$ from \eqref{eq:z0}:
\begin{equation}\label{eq:W_closed}
W = \frac{n}{Z_0\Theta}\,\mathrm{Re}\!\left(\frac{\mathcal{C}\cdot\mathcal{A}}{1+e^{-j\Theta}} + \mathcal{B}\right)
\end{equation}

This is the exact transconductance in closed form.

\subsection{Zero-Crossing Angle}

Within interval $k$, the positive-going current zero crossing ($\mathrm{Im}(z)=0$, $d[\mathrm{Im}]/d\theta > 0$) occurs at:
\begin{equation}\label{eq:sigma_exact}
\sigma_{\text{exact}} = F_n\!\left(\Theta_{k-1} + \arg(z_{k-1}-c_k) + \pi\right)
\end{equation}

selecting the interval where $\arg(z_{k-1}-c_k) + \pi \in [0, \Delta\theta_k]$.

% ==============================================================
\section{Exact Synchronous Rectification and Inversion}
% ==============================================================

The FHA-based feedforward control from \cite{Borisevich2019} inverts the model to find switching parameters $(d, s, \beta)$ from desired commutation angles $(\sigma^*, \delta^*)$. The state-plane framework provides exact closed-form inversions that generalize the FHA results. The full derivation is given in Appendix~A.

\subsection{Synchronous Rectification Condition}

The condition $\delta = 0$ (current crosses zero at the secondary switching edge) has different expressions in the two frameworks:

\begin{equation}\label{eq:sync_fha_vs_exact}
\begin{array}{r@{\;:\quad}l}
\text{FHA \cite{Borisevich2019}} & \cos\beta = G \\[4pt]
\text{Exact (state-plane)} & G\sin\dfrac{\Theta}{2} = \sin\!\left(\dfrac{\Theta}{2} - \hat\beta\right)
\end{array}
\end{equation}

(shown here for the fully-driven case $d = \pi$, $s = 0$). The FHA condition is frequency-independent, while the exact condition depends on $\Theta = \pi/F_n$. Solving the exact condition:

\begin{equation}\label{eq:beta_exact_sync}
\hat\beta = \frac{\Theta}{2} - \arcsin\!\left(G\sin\frac{\Theta}{2}\right)
\end{equation}

At resonance ($F_n = 1$, $\Theta = \pi$): $\hat\beta = \beta = \pi/2 - \arcsin(G) = \arccos(G)$, recovering the FHA result exactly. The physical interpretation: the FHA uses the fundamental component's zero crossing, while the exact condition uses the actual current zero crossing including all harmonics.

\subsection{Exact Inversion: Buck Mode}

Given desired $\sigma^*$, $\delta^*$ (with $\beta = \sigma^* + \delta^*$), the exact primary duty cycle for buck operation ($s = 0$) is:

\begin{equation}\label{eq:inv_buck}
\hat d = \hat\sigma + \frac{\Theta}{2} + \arcsin\!\left(-2G\sin\!\left(\hat\delta - \frac{\Theta}{2}\right) - \sin\!\left(\frac{\Theta}{2} - \hat\sigma\right)\right)
\end{equation}

where $\hat\sigma = \sigma^*/F_n$, $\hat\delta = \delta^*/F_n$. The FHA equivalent from \cite{Borisevich2019} is:
\begin{equation}\label{eq:inv_buck_fha}
d_{FHA} = \arccos(\cos\sigma^* - 2G\cos\delta^*) + \sigma^*
\end{equation}

At resonance ($F_n = 1$) the two are algebraically identical (using $\arccos(x) = \pi/2 - \arcsin(x)$).

\subsection{Exact Inversion: Boost Mode}

When the buck formula yields $d > \pi$, the boost mode ($d = \pi$) is used and the secondary shorting time is:

\begin{equation}\label{eq:inv_boost}
\hat s = \frac{\Theta}{2} - \hat\delta + \arcsin\!\left(\sin\!\left(\frac{\Theta}{2}-\hat\delta\right) - \frac{2}{G}\sin\!\left(\frac{\Theta}{2}-\hat\sigma\right)\right)
\end{equation}

The FHA equivalent from \cite{Borisevich2019} is:
\begin{equation}\label{eq:inv_boost_fha}
s_{FHA} = \arccos\!\left(\frac{2\cos\sigma^*}{G} - \cos\delta^*\right) - \delta^*
\end{equation}

Again identical at $F_n = 1$.

\subsection{Generalization with $s_{add}$}

For low-power operation at fixed maximum frequency, the additional shorting $s_{add}$ is introduced. The buck inversion generalizes to:
\begin{equation}\label{eq:inv_buck_sadd}
\begin{split}
\hat d = \hat\sigma + \frac{\Theta}{2} + \arcsin\!\Big(&{-}G\sin\!\left(\hat\delta{-}\tfrac{\Theta}{2}\right) \\
&{-} G\sin\!\left(\hat\delta{+}\hat s_{add}{-}\tfrac{\Theta}{2}\right) - \sin\!\left(\tfrac{\Theta}{2}{-}\hat\sigma\right)\Big)
\end{split}
\end{equation}

reducing to \eqref{eq:inv_buck} for $s_{add} = 0$.

\subsection{Aggregated Switching Parameter $q$}

Following \cite{Borisevich2019}, the buck and boost modes are combined into a single continuous variable with $s_{add}$ accommodation:
\begin{equation}\label{eq:q_def}
\begin{pmatrix} d \\ s \end{pmatrix} = \begin{cases} \begin{pmatrix} q \\ s_{add} \end{pmatrix} & q \le \pi \;\text{(buck)} \\[4pt] \begin{pmatrix} \pi \\ q - \pi \end{pmatrix} & q > \pi \;\text{(boost)} \end{cases}
\end{equation}

In buck mode, the inversion \eqref{eq:inv_buck_sadd} directly gives $q = d$ (with $s = s_{add}$ as an external parameter). In boost mode, $q = s + \pi$ where $s$ is obtained from \eqref{eq:inv_boost}. Since the boost formula \eqref{eq:inv_boost} already computes $s$ for $d = \pi$, no separate inversion is needed---the same formula applies with $q = F_n\hat s + \pi$.

Fig.~\ref{fig:q_comparison} compares $q(G)$ from the exact inversion with the FHA result across several normalized frequencies. At $F_n = 1$ the curves coincide; as $F_n$ increases the exact $q$ deviates progressively from the frequency-independent FHA curve.

\begin{figure}[t]
\centering
% this figure is generated by sim/plot_q_comparison.py script
\includegraphics[width=\columnwidth]{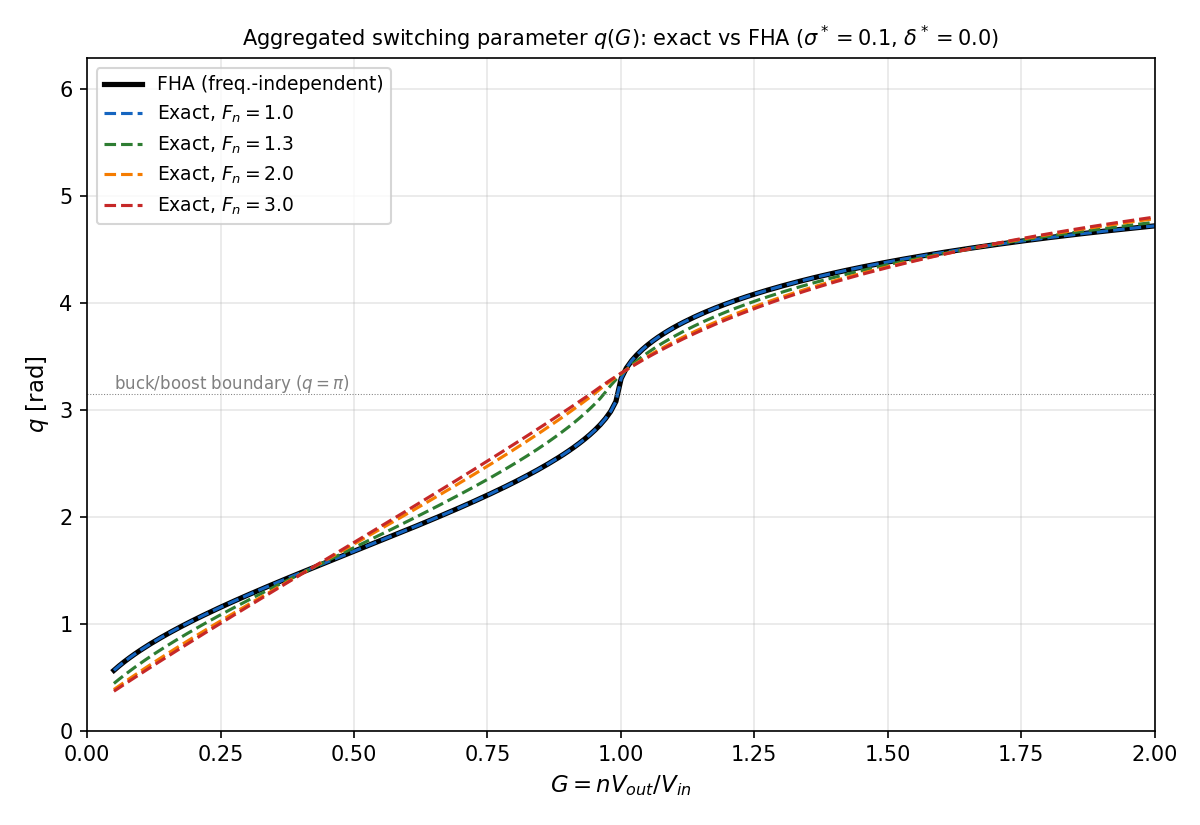}
\caption{Aggregated switching parameter $q(G)$: exact state-plane inversion (dashed) vs.\ FHA (solid black) for $\sigma^*=0.1$, $\delta^*=0$. At resonance ($F_n=1$) both coincide; deviation grows with $F_n$.}
\label{fig:q_comparison}
\end{figure}

% ==============================================================
\section{Controller Architecture}
% ==============================================================

The control system uses the same parallel nonlinear compensation structure as \cite{Borisevich2019} (Fig.~\ref{fig:Freq_Ctrl_parallel}): the feedforward maps \eqref{eq:inv_buck}--\eqref{eq:inv_boost} compute nominal $q$ and $\beta$ from $(\sigma^*, \delta^*, G, F_n)$, with PI controllers adding corrections for robustness. The outer loop adjusts $\omega$ (and $s_{add}$ in low-power mode) to regulate $W$ using the exact transconductance \eqref{eq:W_closed}.

\begin{figure}[t]
%\input{Freq_Ctrl_parallel.TpX}  % TODO: adapt figure from original paper
% no generating script found in sim; static/imported artwork
\includegraphics[width=\columnwidth]{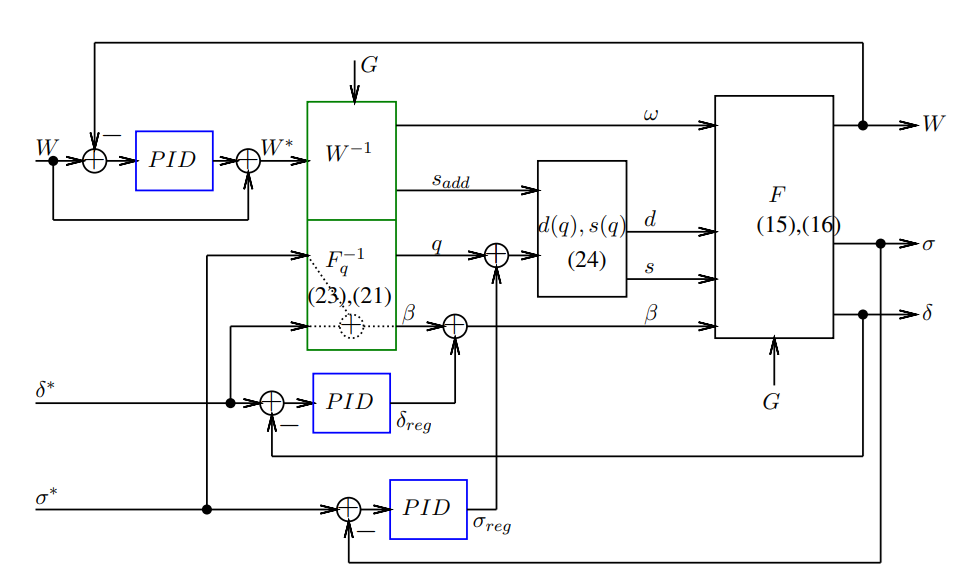}
\centering
\caption{Controller architecture with commutation parameters control $\sigma$ and $\delta$, and output power control $W$.}
\label{fig:Freq_Ctrl_parallel}
\end{figure}

The key difference from the FHA-based architecture: the inversion maps now operate on \textit{resonant-time quantities} ($\hat d$, $\hat s$, $\hat\beta$) which are physical time intervals (in units of $1/\omega_0$), not switching-frequency angles. These quantities depend on $F_n$, coupling the inner commutation loop to the outer frequency loop. In practice, the feedforward look-up table is indexed by $(G, F_n)$ rather than $G$ alone, reflecting the improved physical modeling provided by the state-plane formalism.

% ==============================================================
\section{Numerical Results}
% ==============================================================

The model is evaluated for a converter with $L = 31\,\mu$H, $C = 8.2\,$nF, $n = 2.2$ ($f_0 = 315.7\,$kHz, $Z_0 = 61.5\,\Omega$).

\subsection{Transconductance Comparison}

Fig.~\ref{fig:W_105} and Fig.~\ref{fig:W_20} compare $W_{\text{exact}}$ from \eqref{eq:W_closed} with $W_{FHA}$ from \eqref{eq:FHA_model} at $F_n = 1.05$ (near resonance) and $F_n = 2.0$ (far from resonance).

\begin{figure}[t]
\centering
% this figure is generated by sim/state_plane_paper_plots_v2.py script
\includegraphics[width=\columnwidth]{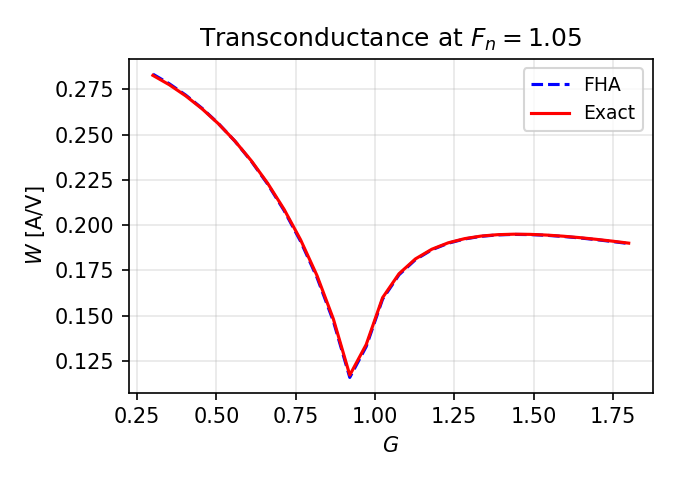}
\caption{Transconductance comparison at $F_n = 1.05$: exact and FHA nearly coincide.}
\label{fig:W_105}
\end{figure}

\begin{figure}[t]
\centering
% this figure is generated by sim/state_plane_paper_plots_v2.py script
\includegraphics[width=\columnwidth]{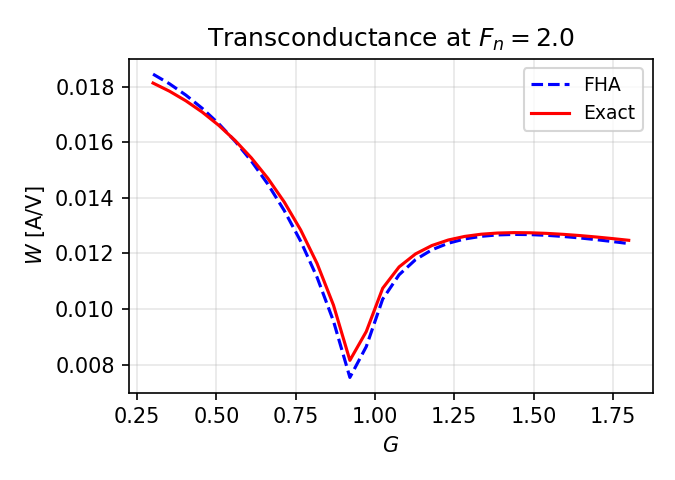}
\caption{Transconductance comparison at $F_n = 2.0$: small but visible deviation appears.}
\label{fig:W_20}
\end{figure}

\subsection{Inversion Accuracy}

Table~\ref{tab:inversion} compares the exact inversion formulas with FHA. The exact formulas achieve the target $\sigma^*$ precisely, while FHA produces errors of 3--72\% depending on operating point.

\begin{table}[t]
\centering
\caption{Inversion comparison: exact formula vs.\ FHA ($\sigma$ achieved)}
\label{tab:inversion}
\small
\begin{tabular}{cccc|cc}
\toprule
$\sigma^*$ & $\delta^*$ & $G$ & $F_n$ & $\sigma$ at $d_{\text{exact}}$ & $\sigma$ at $d_{FHA}$ \\
\midrule
0.1 & 0.0 & 0.7 & 1.50 & 0.100 & 0.028 \\
0.2 & 0.0 & 0.7 & 1.50 & 0.200 & 0.110 \\
0.1 & 0.1 & 0.7 & 1.50 & 0.100 & 0.059 \\
0.3 & 0.0 & 0.5 & 1.30 & 0.300 & 0.251 \\
0.1 & 0.0 & 0.7 & 1.05 & 0.100 & 0.079 \\
\bottomrule
\end{tabular}
\end{table}

\subsection{Commutation Angle Prediction}

Fig.~\ref{fig:sigma_error} shows the relative $\sigma$ prediction error. The FHA consistently overestimates $\sigma$ by 5--20\%, meaning the actual ZVS margin is lower than predicted.

\begin{figure}[t]
\centering
% this figure is generated by sim/state_plane_paper_plots_v2.py script
\includegraphics[width=\columnwidth]{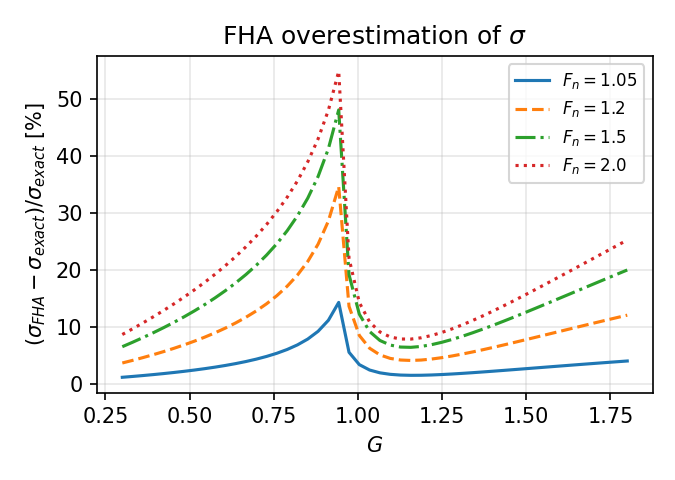}
\caption{FHA overestimation of $\sigma$ relative to exact state-plane value.}
\label{fig:sigma_error}
\end{figure}

\subsection{Partial Rectification Mode}

Fig.~\ref{fig:partial_rect} shows FHA error growth with $s_{add}$, reaching 15\% for $s_{add} = 2.5\,$rad.

\begin{figure}[t]
\centering
% this figure is generated by sim/state_plane_paper_plots_v2.py script
\includegraphics[width=\columnwidth]{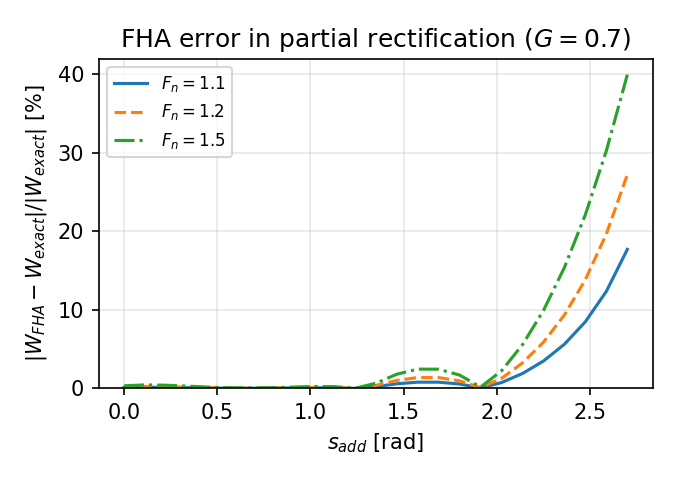}
\caption{FHA transconductance error in partial rectification mode ($G=0.7$, $\sigma^*=0.1$).}
\label{fig:partial_rect}
\end{figure}

% ==============================================================
\section{Conclusion}
% ==============================================================

Exact closed-form feedforward inversion maps have been derived for the dual-bridge series resonant converter using state-plane trajectory analysis. The main contributions are:

\begin{enumerate}
\item Exact inversion formulas \eqref{eq:inv_buck}--\eqref{eq:inv_boost} for computing switching parameters $(d, s, \beta)$ from desired commutation angles $(\sigma^*, \delta^*)$. These have a structure analogous to the FHA trigonometric equations from \cite{Borisevich2019} but incorporate the normalized frequency through $\Theta = \pi/F_n$ and commutation angles normalization.

\item A rigorous proof that the FHA inversion is the exact state-plane inversion evaluated at resonance ($F_n = 1$). This establishes the FHA as a special case rather than an independent approximation.

\item For above-resonance operation ($F_n > 1$), the exact inversions eliminate the 5--72\% commutation angle errors of FHA-based feedforward, enabling precise ZVS margin control.

\item The generalization to partial rectification mode \eqref{eq:inv_buck_sadd} retains the closed-form structure, providing exact low-power control maps.

\item The controller architecture mirrors the parallel nonlinear compensation from \cite{Borisevich2019}, with feedforward maps now operating on resonant-time quantities ($\hat d$, $\hat s$, $\hat\beta$ in units of $1/\omega_0$) that couple the commutation and frequency loops.
\end{enumerate}

The state-plane framework (complex state variable, arc composition recurrence, half-wave symmetry) follows the established approach of \cite{Oruganti1985, Oruganti1988}, extended here to the DB SRC topology with full $(d, s, \beta, \omega)$ modulation. The novel contribution is the derivation of exact invertible feedforward maps from this framework and the explicit connection to the FHA results.

% ==============================================================
% REFERENCES
% ==============================================================

% ==============================================================
% APPENDIX: Full Mathematical Derivation (single-column for readability)
% ==============================================================
\onecolumn
\appendices
\section{Exact State-Plane Analysis of Dual-Bridge SRC}

\subsection{Normalized state-plane formulation}

Consider the series LC resonant tank driven by the net voltage $v(t) = v_{in}(t) - n \cdot v_{out}(t)$. The state variables are the inductor current $i_L(t)$ and the capacitor voltage $v_C(t)$:

\begin{equation}\label{eq:sp_dynamics}
L \frac{d i_L}{dt} = v(t) - v_C(t), \quad C \frac{d v_C}{dt} = i_L(t)
\end{equation}

With the resonant frequency $\omega_0 = 1/\sqrt{LC}$, characteristic impedance $Z_0 = \sqrt{L/C}$, and normalized coordinates:

\begin{equation}\label{eq:sp_normalization}
x = \frac{v_C}{V_{in}}, \quad y = \frac{Z_0 \, i_L}{V_{in}}, \quad \theta = \omega_0 t
\end{equation}

the dynamics become:

\begin{equation}\label{eq:sp_norm_dynamics}
\frac{dx}{d\theta} = y, \quad \frac{dy}{d\theta} = c - x
\end{equation}

where $c = V_n/V_{in}$ is the normalized applied voltage (constant within each switching interval).

\subsection{Complex state variable and rotation}

Introduce the complex state variable:

\begin{equation}\label{eq:sp_complex_state}
z = x + jy
\end{equation}

The solution of \eqref{eq:sp_norm_dynamics} for constant applied voltage $c$ with initial condition $z_0$ is:

\begin{equation}\label{eq:sp_complex_solution}
z(\theta) = c + (z_0 - c)\,e^{-j\theta}
\end{equation}

This represents a clockwise circular arc in the complex plane, centered at $c$ (on the real axis), with radius $|z_0 - c|$. The physical interpretation is: the normalized capacitor voltage oscillates around the applied voltage level $c$, with the normalized current as the imaginary part.

\subsection{Switching pattern and interval decomposition}

The positive half-cycle of the converter (one half switching period, $t' \in [0,\pi]$ in switching-frequency coordinates) is decomposed into $N$ intervals with constant applied voltages $c_1, c_2, \ldots, c_N$ and angular durations $\alpha_1, \alpha_2, \ldots, \alpha_N$ (in switching-frequency radians, satisfying $\sum_k \alpha_k = \pi$).

Converting to resonant-time coordinates, the durations become $\Delta\theta_k = \alpha_k / F_n$ where $F_n = \omega/\omega_0$ is the normalized switching frequency. The half-period in resonant-time is $\Theta = \pi / F_n$.

For the DB SRC with parameters $(d, s, \beta, G)$, two principal orderings arise depending on the relative position of the primary turn-off edge $d$ and the secondary active-to-short transition at $s + \beta$. In both cases $N = 4$ intervals span the positive half-cycle:

\textbf{Case 1:} $d \ge s + \beta$ (primary pulse encompasses secondary shorting interval):
\begin{equation}\label{eq:sp_case1}
\begin{array}{c|c|c}
k & c_k / V_{in} & \alpha_k \\
\hline
1 & 1 + G & \beta \\
2 & 1 & s \\
3 & 1 - G & d - s - \beta \\
4 & -G & \pi - d
\end{array}
\end{equation}

\textbf{Case 2:} $d < s + \beta$ (primary pulse ends before secondary becomes active):
\begin{equation}\label{eq:sp_case2}
\begin{array}{c|c|c}
k & c_k / V_{in} & \alpha_k \\
\hline
1 & 1 + G & \beta \\
2 & 1 & d - \beta \\
3 & 0 & s - d + \beta \\
4 & -G & \pi - s - \beta
\end{array}
\end{equation}

In both cases, the first interval has the secondary in negative-active state (carry-over from previous half-cycle), and the last interval has the primary off with secondary in positive-active state. Intervals with $\alpha_k = 0$ (e.g., when $d = \pi$ or $s = 0$) are simply omitted, reducing $N$ accordingly.

All subsequent derivations use the values of $c_k$ and $\alpha_k$ (or equivalently $\Delta\theta_k = \alpha_k/F_n$) taken from the applicable table \eqref{eq:sp_case1} or \eqref{eq:sp_case2}, with $N$ being the number of non-zero-duration intervals.

\subsection{Composition of rotations}

Starting from initial state $z_0$, the state after interval $k$ is:

\begin{equation}\label{eq:sp_one_rotation}
z \mapsto c_k + (z - c_k)\,e^{-j\Delta\theta_k}
\end{equation}

Composing all $N$ intervals sequentially:

\begin{equation}\label{eq:sp_composition}
z_N = \mathcal{F}(z_0) = e^{-j\Theta} z_0 + \mathcal{C}
\end{equation}

where $\Theta = \sum_{k=1}^N \Delta\theta_k = \pi/F_n$ is the half-period, and the complex offset $\mathcal{C}$ is computed by the recurrence:

\begin{equation}\label{eq:sp_recurrence}
\mathcal{C}_0 = 0, \quad \mathcal{C}_k = (\mathcal{C}_{k-1} + c_k)\,e^{-j\Delta\theta_k} - c_k, \quad \mathcal{C} = \mathcal{C}_N
\end{equation}

Figure~\ref{fig:sp_composition} illustrates the composition of rotations for a representative operating point. Each colored arc corresponds to one switching interval, centered at $c_k$ on the real axis. The state advances clockwise along each arc for an angular extent of $\Delta\theta_k$ radians, then transitions to the next arc centered at $c_{k+1}$.

\begin{figure}[h]
\centering
% this figure is generated by sim/state_plane_illustrations.py script
\includegraphics[width=0.85\textwidth]{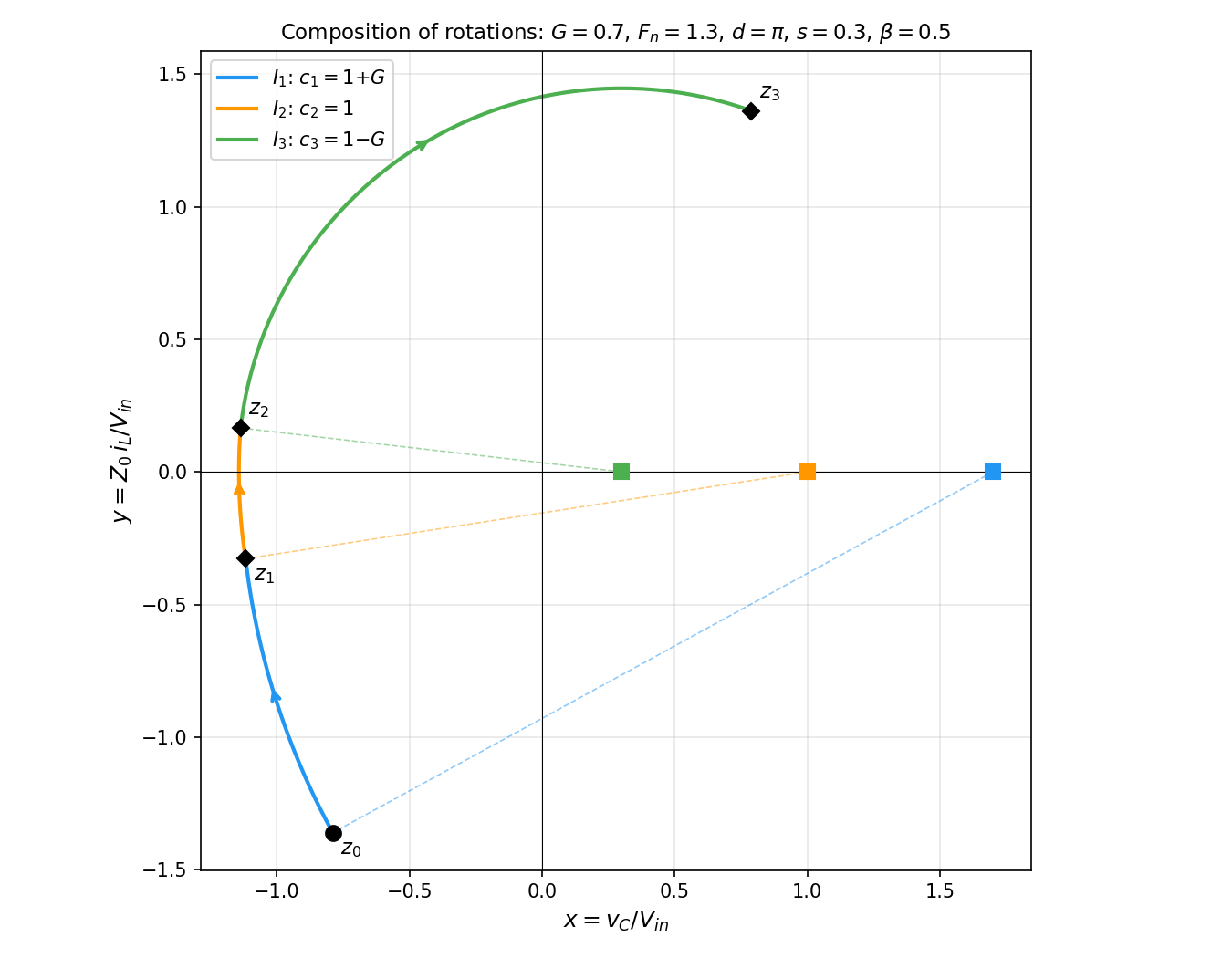}
\caption{Composition of circular arcs in the state plane. Each interval $I_k$ traces a clockwise arc about center $c_k$ (square markers). The switching instants $z_0, z_1, \ldots, z_N$ are marked with circles/diamonds.}
\label{fig:sp_composition}
\end{figure}

\subsection{Half-wave symmetry and steady-state solution}

Under half-wave symmetric modulation, the negative half-cycle is the exact negation of the positive half-cycle. The periodicity condition over one full cycle therefore reduces to:

\begin{equation}\label{eq:sp_half_wave_symmetry}
z(\Theta) = -z(0)
\end{equation}

This condition is illustrated in Figure~\ref{fig:sp_halfwave}: the trajectory in the negative half-cycle is the point-symmetric image (through the origin) of the positive half-cycle trajectory.

\begin{figure}[h]
\centering
% this figure is generated by sim/state_plane_illustrations.py script
\includegraphics[width=0.85\textwidth]{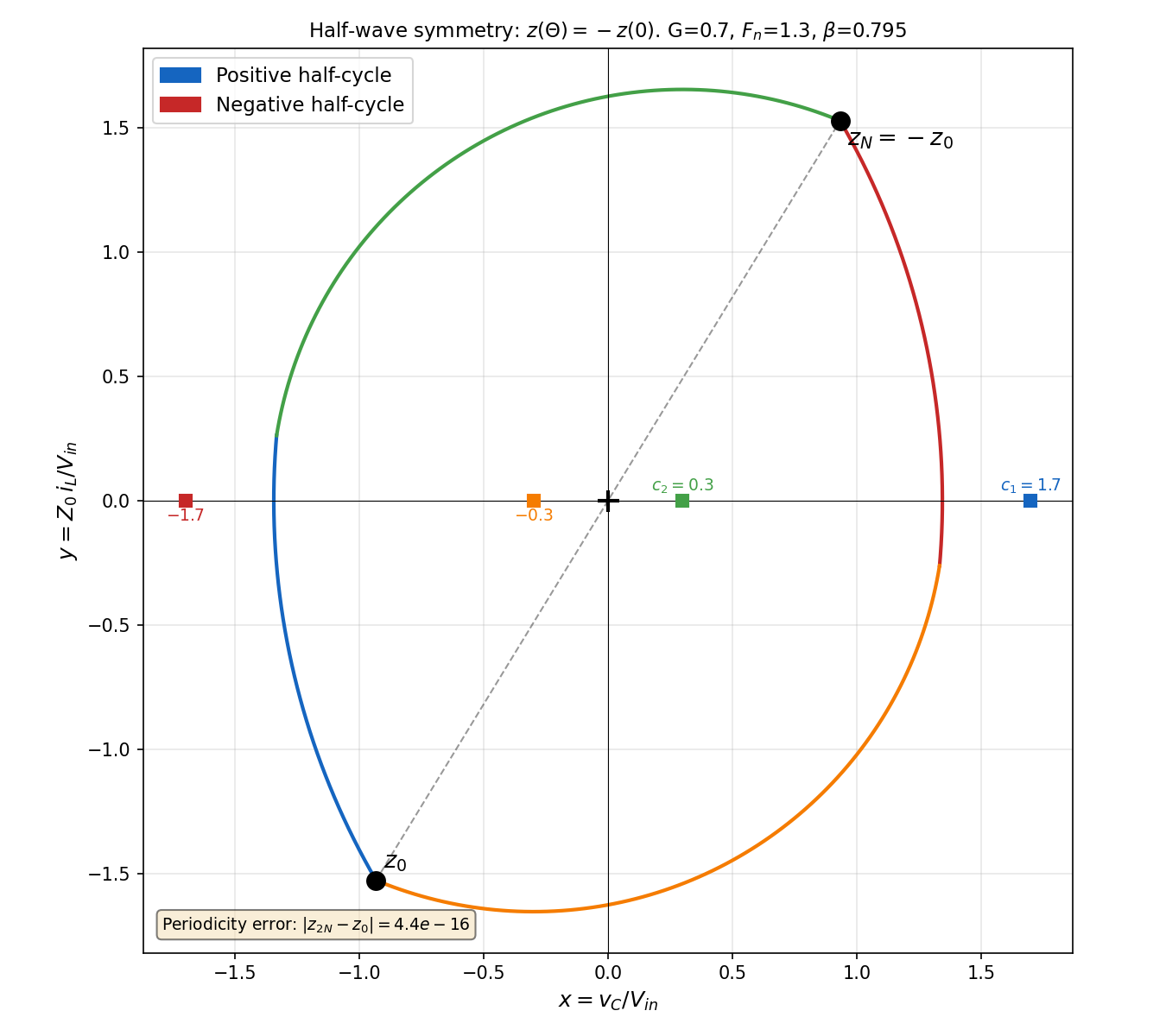}
\caption{Half-wave symmetry in the state plane. The positive half-cycle (blue/green arcs) starts at $z_0$ and ends at $z_N = -z_0$. The negative half-cycle (red/orange arcs) traverses the point-symmetric trajectory, returning to $z_0$.}
\label{fig:sp_halfwave}
\end{figure}

Substituting \eqref{eq:sp_composition}:

\begin{equation}
e^{-j\Theta}\,z_0 + \mathcal{C} = -z_0
\end{equation}

Note: the recurrence \eqref{eq:sp_recurrence} computes $\mathcal{C} = -\mathbf{b}$, where $\mathbf{b}$ is the true affine offset of the composed map. This can be verified by observing that propagating $z = 0$ through the physical arc maps yields $\mathbf{b}$, while the recurrence starting from 0 yields $-\mathbf{b}$. The composed map is therefore $z_N = e^{-j\Theta}z_0 - \mathcal{C}$, and the half-wave condition $z_N = -z_0$ gives:

\begin{equation}
-z_0 = e^{-j\Theta}\,z_0 - \mathcal{C} \quad \Longrightarrow \quad z_0(1 + e^{-j\Theta}) = \mathcal{C}
\end{equation}

Solving for the initial state:

\begin{equation}\label{eq:sp_z0_solution}
\boxed{z_0 = \frac{\mathcal{C}}{1 + e^{-j\Theta}}}
\end{equation}

Using the identity $1 + e^{-j\Theta} = 2\cos(\Theta/2)\,e^{-j\Theta/2}$:

\begin{equation}\label{eq:sp_z0_polar}
z_0 = \frac{\mathcal{C}}{2\cos(\Theta/2)}\,e^{j\Theta/2}
\end{equation}

This is the complete closed-form solution for the steady-state initial conditions. The real and imaginary parts give:

\begin{equation}\label{eq:sp_x0_y0}
x_0 = \mathrm{Re}(z_0), \quad y_0 = \mathrm{Im}(z_0)
\end{equation}

\subsection{Explicit expansion of $\mathcal{C}$}

Expanding the recurrence \eqref{eq:sp_recurrence}:

\begin{equation}\label{eq:sp_C_expanded}
\mathcal{C} = \sum_{k=1}^N c_k \left(e^{-j\Delta\theta_k} - 1\right) \prod_{m=k+1}^N e^{-j\Delta\theta_m}
\end{equation}

Using $e^{-j\alpha} - 1 = -2j\sin(\alpha/2)\,e^{-j\alpha/2}$ and defining $\Phi_k = \sum_{m=k+1}^N \Delta\theta_m$ (remaining angle after interval $k$):

\begin{equation}\label{eq:sp_C_trig}
\mathcal{C} = -2j \sum_{k=1}^N c_k \sin\frac{\Delta\theta_k}{2}\, e^{-j(\Delta\theta_k/2 + \Phi_k)}
\end{equation}

For the special case of equal-angle normalization ($\Delta\theta_k = \alpha_k/F_n$):

\begin{equation}\label{eq:sp_C_explicit}
\mathcal{C} = -\frac{2j}{F_n} \sum_{k=1}^N c_k \sin\frac{\alpha_k}{2F_n}\, e^{-j(\alpha_k/2 + \sum_{m=k+1}^N \alpha_m)/F_n}
\end{equation}

\subsection{State at intermediate switching instants}

The state at the end of interval $k$ (needed for output current and zero-crossing calculations) is obtained by partial composition. Define:

\begin{equation}\label{eq:sp_partial_state}
z_k = e^{-j\Theta_k}\,z_0 - \mathcal{C}_k
\end{equation}

where $\Theta_k = \sum_{m=1}^k \Delta\theta_m$ and $\mathcal{C}_k$ is the partial offset from the recurrence \eqref{eq:sp_recurrence}. The minus sign reflects the fact that $\mathcal{C}_k$ is the negated affine offset of the composed map through the first $k$ intervals.

\subsection{Exact output transconductance}\label{sec:sp_W_exact}

The output current is determined by the power delivered through the secondary bridge:

\begin{equation}
I_{out} = \frac{\langle v_{out}(t) \cdot i_L(t) \rangle}{n V_{out}}
\end{equation}

Over the positive half-cycle, the integral of current weighted by the rectification function is:

\begin{equation}
\int_0^\Theta r(\theta)\,y(\theta)\,d\theta = \sum_{k=1}^N r_k \int_{\Theta_{k-1}}^{\Theta_k} y(\theta)\,d\theta = \sum_{k=1}^N r_k (x_k - x_{k-1})
\end{equation}

where $r_k$ is the rectification sign ($+1$ for secondary active positive, $-1$ for secondary active negative, $0$ for shorted).

By half-wave symmetry, the negative half-cycle contribution equals the positive half-cycle contribution (both $v_{out}$ and $i_L$ negate, so their product is unchanged). Therefore:

\begin{equation}\label{eq:sp_W_half_cycle}
W = \frac{n}{Z_0 \cdot \Theta} \sum_{k=1}^N r_k\,\mathrm{Re}(z_k - z_{k-1})
\end{equation}

where $\mathrm{Re}(z_k - z_{k-1}) = x_k - x_{k-1}$ is the capacitor voltage change in interval $k$.

Now, using $z_k - z_{k-1} = (z_{k-1} - c_k)(e^{-j\Delta\theta_k} - 1)$ from \eqref{eq:sp_one_rotation}:

\begin{equation}\label{eq:sp_delta_z}
z_k - z_{k-1} = (z_{k-1} - c_k)(e^{-j\Delta\theta_k} - 1)
\end{equation}

Taking the real part:

\begin{equation}\label{eq:sp_delta_x}
x_k - x_{k-1} = (x_{k-1} - c_k)(\cos\Delta\theta_k - 1) + y_{k-1}\sin\Delta\theta_k
\end{equation}

Substituting into \eqref{eq:sp_W_half_cycle}:

\begin{equation}\label{eq:sp_W_expanded}
W = \frac{n}{Z_0 \Theta} \sum_{k=1}^N r_k \left[ (x_{k-1} - c_k)(\cos\Delta\theta_k - 1) + y_{k-1}\sin\Delta\theta_k \right]
\end{equation}

\subsubsection{Compact complex form of $W$}

Using the complex delta notation $\Delta z_k = z_k - z_{k-1} = (z_{k-1} - c_k)(e^{-j\Delta\theta_k} - 1)$, the transconductance can be written as:

\begin{equation}\label{eq:sp_W_compact}
W = \frac{n}{Z_0\Theta}\,\mathrm{Re}\!\left(\sum_{k=1}^N r_k\,\Delta z_k\right) = \frac{n}{Z_0\Theta}\,\mathrm{Re}\!\left(\sum_{k=1}^N r_k (z_{k-1} - c_k)(e^{-j\Delta\theta_k} - 1)\right)
\end{equation}

Since $z_{k-1} = e^{-j\Theta_{k-1}} z_0 - \mathcal{C}_{k-1}$ (the intermediate state, with $\mathcal{C}_{k-1}$ from the recurrence), this is linear in $z_0$:

\begin{equation}\label{eq:sp_W_linear_z0}
W = \frac{n}{Z_0\Theta}\,\mathrm{Re}\!\left( z_0 \cdot \mathcal{A} + \mathcal{B} \right)
\end{equation}

where:
\begin{equation}\label{eq:sp_AB_coeff}
\begin{split}
\mathcal{A} &= \sum_{k=1}^N r_k\, e^{-j\Theta_{k-1}} (e^{-j\Delta\theta_k} - 1) \\
\mathcal{B} &= \sum_{k=1}^N r_k\, (-\mathcal{C}_{k-1} - c_k)(e^{-j\Delta\theta_k} - 1)
\end{split}
\end{equation}

Substituting $z_0 = \mathcal{C}/(1 + e^{-j\Theta})$ from \eqref{eq:sp_z0_solution}:

\begin{equation}\label{eq:sp_W_final_closed}
\boxed{W = \frac{n}{Z_0\Theta}\,\mathrm{Re}\!\left( \frac{\mathcal{C} \cdot \mathcal{A}}{1 + e^{-j\Theta}} + \mathcal{B} \right)}
\end{equation}

This is the \textbf{exact closed-form transconductance}, expressed entirely through the switching parameters via $\mathcal{C}$, $\mathcal{A}$, $\mathcal{B}$.

\subsection{Fully-driven converter analysis}

For the fully-driven case ($d = \pi$) with the interval structure from Case~1 (with $d \ge s + \beta$, which gives $d - s - \beta = \pi - s - \beta$ and last interval duration $\pi - d = 0$, so we reduce to 3 intervals):

\begin{equation}\label{eq:sp_3int_full_drive}
\begin{array}{c|c|c|c}
k & c_k/V_{in} & \alpha_k & r_k \\
\hline
1 & 1 + G & \beta & -1 \\
2 & 1 & s & 0 \\
3 & 1 - G & \pi - s - \beta & +1
\end{array}
\end{equation}

The rectification signs are: $r_1 = -1$ (secondary active negative, carry-over from previous half-cycle), $r_2 = 0$ (shorted), $r_3 = +1$ (secondary active positive). The transconductance becomes:

\begin{equation}\label{eq:sp_W_3int}
W = \frac{n}{Z_0 \Theta}\left[ -(x_1 - x_0) + (x_3 - x_2) \right]
\end{equation}

Using $x_3 = -x_0$ (half-wave symmetry) and noting that $x_1 - x_0 + x_2 - x_1 + x_3 - x_2 = -2x_0$:

\begin{equation}
(x_3 - x_2) = -2x_0 - (x_1 - x_0) - (x_2 - x_1)
\end{equation}

Therefore:

\begin{equation}\label{eq:sp_W_3int_v2}
W = \frac{n}{Z_0 \Theta}\left[ -2(x_1 - x_0) - (x_2 - x_1) - 2x_0 \right]
\end{equation}

Or equivalently, combining with the half-wave symmetry $x_3 = -x_0$:

\begin{equation}\label{eq:sp_W_3int_v3}
W = \frac{n}{Z_0 \Theta}\left[-2x_0 - 2(x_1 - x_0) - (x_2 - x_1)\right] = \frac{n}{Z_0\Theta}\left[-2x_1 - (x_2 - x_1)\right]
\end{equation}

Substituting the voltage changes using \eqref{eq:sp_delta_x}:

\begin{equation}\label{eq:sp_W_full_analytic}
\begin{split}
W = \frac{n}{Z_0\Theta} \bigg[ &-2\,\mathrm{Re}\!\left((z_0 - c_1)(e^{-j\Delta\theta_1} - 1)\right) \\
&\quad -\,\mathrm{Re}\!\left((z_1 - c_2)(e^{-j\Delta\theta_2} - 1)\right) \bigg]
\end{split}
\end{equation}

where $z_0$ is given by \eqref{eq:sp_z0_solution}, $c_1 = 1+G$, $c_2 = 1$, $\Delta\theta_1 = \beta/F_n$, $\Delta\theta_2 = s/F_n$, and $z_1 = c_1 + (z_0 - c_1)e^{-j\Delta\theta_1}$.

\subsubsection{Exact $W$ for the fully-driven converter ($d = \pi$, $s \ge 0$)}

For $d = \pi$, the interval structure \eqref{eq:sp_3int_full_drive} gives 3 intervals with $\Delta\theta_1 = \beta/F_n$, $\Delta\theta_2 = s/F_n$, $\Delta\theta_3 = (\pi-s-\beta)/F_n$, and rectification signs $r_1 = -1$, $r_2 = 0$, $r_3 = +1$.

Define shorthand angles $\hat\beta = \beta/F_n$, $\hat s = s/F_n$, $\Theta = \pi/F_n$. The offset $\mathcal{C}$ from the recurrence \eqref{eq:sp_recurrence} with normalized centers $c_1 = 1+G$, $c_2 = 1$, $c_3 = 1-G$:

\begin{equation}\label{eq:sp_C_step1}
\mathcal{C}_1 = (0 + c_1)e^{-j\hat\beta} - c_1 = c_1(e^{-j\hat\beta} - 1)
\end{equation}
\begin{equation}\label{eq:sp_C_step2}
\mathcal{C}_2 = (\mathcal{C}_1 + c_2)e^{-j\hat s} - c_2 = \mathcal{C}_1 e^{-j\hat s} + c_2(e^{-j\hat s} - 1)
\end{equation}
\begin{equation}\label{eq:sp_C_step3}
\mathcal{C} = \mathcal{C}_3 = (\mathcal{C}_2 + c_3)e^{-j(\Theta-\hat\beta-\hat s)} - c_3 = \mathcal{C}_2\, e^{-j(\Theta-\hat\beta-\hat s)} + c_3(e^{-j(\Theta-\hat\beta-\hat s)} - 1)
\end{equation}

Expanding fully:
\begin{equation}\label{eq:sp_C_full}
\mathcal{C} = c_1(e^{-j\hat\beta}-1)e^{-j(\Theta-\hat\beta)} + c_2(e^{-j\hat s}-1)e^{-j(\Theta-\hat\beta-\hat s)} + c_3(e^{-j(\Theta-\hat\beta-\hat s)}-1)
\end{equation}

Substituting $c_1 = 1+G$, $c_2 = 1$, $c_3 = 1-G$ and using $e^{-j\alpha} - 1 = -2j\sin(\alpha/2)e^{-j\alpha/2}$:

\begin{equation}\label{eq:sp_C_trig_full}
\begin{split}
\mathcal{C} = -2j\bigg[ &(1+G)\sin\frac{\hat\beta}{2}\,e^{-j(\Theta - \hat\beta/2)} \\
&+ \sin\frac{\hat s}{2}\,e^{-j(\Theta - \hat\beta - \hat s/2)} \\
&+ (1-G)\sin\frac{\Theta-\hat\beta-\hat s}{2}\,e^{-j(\Theta-\hat\beta-\hat s)/2} \bigg]
\end{split}
\end{equation}

The coefficients $\mathcal{A}$ and $\mathcal{B}$ from \eqref{eq:sp_AB_coeff} with $r_1=-1$, $r_2=0$, $r_3=+1$:

\begin{equation}\label{eq:sp_A_full}
\mathcal{A} = -(e^{-j\hat\beta}-1) + e^{-j(\hat\beta+\hat s)}(e^{-j(\Theta-\hat\beta-\hat s)}-1)
\end{equation}

Simplifying:
\begin{equation}
\mathcal{A} = 1 - e^{-j\hat\beta} + e^{-j\Theta} - e^{-j(\hat\beta+\hat s)} = (1-e^{-j\hat\beta}) - e^{-j(\hat\beta+\hat s)}(1 - e^{-j(\Theta-\hat\beta-\hat s)})
\end{equation}

And:
\begin{equation}\label{eq:sp_B_full}
\mathcal{B} = -(\mathcal{C}_0 - c_1)(e^{-j\hat\beta}-1) + (\mathcal{C}_2 - c_3)(e^{-j(\Theta-\hat\beta-\hat s)}-1)
\end{equation}

with $\mathcal{C}_0 = 0$.

\subsubsection{Special case $s = 0$, $d = \pi$ (pure synchronous rectification)}

With $s = 0$, only 2 distinct intervals remain in the positive half-cycle:
\begin{equation}
\begin{array}{c|c|c|c}
k & c_k/V_{in} & \alpha_k & r_k \\
\hline
1 & 1 + G & \beta & -1 \\
2 & 1 - G & \pi - \beta & +1
\end{array}
\end{equation}

The offset $\mathcal{C}$ simplifies to:
\begin{equation}\label{eq:sp_C_s0}
\mathcal{C} = (1+G)(e^{-j\hat\beta}-1)e^{-j(\Theta-\hat\beta)} + (1-G)(e^{-j(\Theta-\hat\beta)}-1)
\end{equation}

Expanding and collecting:
\begin{equation}\label{eq:sp_C_s0_expanded}
\mathcal{C} = (1+G)e^{-j\Theta} - (1+G)e^{-j(\Theta-\hat\beta)} + (1-G)e^{-j(\Theta-\hat\beta)} - (1-G)
\end{equation}
\begin{equation}
= (1+G)e^{-j\Theta} - 2Ge^{-j(\Theta-\hat\beta)} - (1-G)
\end{equation}

The initial state from \eqref{eq:sp_z0_solution}:
\begin{equation}\label{eq:sp_z0_s0}
z_0 = \frac{\mathcal{C}}{1 + e^{-j\Theta}} = \frac{(1+G)e^{-j\Theta} - 2Ge^{-j(\Theta-\hat\beta)} - (1-G)}{1 + e^{-j\Theta}}
\end{equation}

For the transconductance with $r_1 = -1$, $r_2 = +1$ (using \eqref{eq:sp_W_half_cycle}):
\begin{equation}\label{eq:sp_W_s0}
W = \frac{n}{Z_0\Theta}\left[-(x_1-x_0) + (x_2-x_1)\right] = \frac{n}{Z_0\Theta}\left[-2(x_1-x_0) + (x_2-x_0)\right]
\end{equation}

Using $x_2 - x_0 = \mathrm{Re}(z_2 - z_0) = \mathrm{Re}(-2z_0) = -2x_0$ (half-wave symmetry):
\begin{equation}\label{eq:sp_W_s0_v2}
W = \frac{n}{Z_0\Theta}\left[-2\,\mathrm{Re}(\Delta z_1) - 2x_0\right]
\end{equation}

where $\Delta z_1 = (z_0 - c_1)(e^{-j\hat\beta}-1)$ with $c_1 = 1+G$.

Substituting $z_0$ from \eqref{eq:sp_z0_s0}:
\begin{equation}\label{eq:sp_W_s0_final}
W = \frac{n}{Z_0\Theta}\,\mathrm{Re}\!\left[\frac{-2\mathcal{C}}{1+e^{-j\Theta}} + 2\left(\frac{\mathcal{C}}{1+e^{-j\Theta}} - (1+G)\right)(e^{-j\hat\beta}-1) \right]
\end{equation}

which after simplification gives:
\begin{equation}\label{eq:sp_W_s0_closed}
\boxed{W\big|_{s=0} = \frac{2n}{Z_0\Theta}\,\mathrm{Re}\!\left[\frac{\mathcal{C}}{1+e^{-j\Theta}}\left(-1 + (e^{-j\hat\beta}-1)\right) - (1+G)(e^{-j\hat\beta}-1) \right]}
\end{equation}

with $\mathcal{C}$ from \eqref{eq:sp_C_s0_expanded}.

\subsection{Explicit expansion of all coefficients}

In this section the general formulas are expanded explicitly for both switching pattern cases. Define shorthand notation for normalized angles: $\hat\alpha = \alpha/F_n$ for any switching angle $\alpha$, so that $\hat\beta = \beta/F_n$, $\hat s = s/F_n$, $\hat d = d/F_n$, and $\Theta = \pi/F_n$.

\subsubsection{Case 1: $d \ge s + \beta$}

The intervals are $(c_k, \alpha_k) = \{(1{+}G,\,\beta),\; (1,\,s),\; (1{-}G,\,d{-}s{-}\beta),\; (-G,\,\pi{-}d)\}$ with rectification signs $r_k = \{-1,\, 0,\, +1,\, +1\}$.

\paragraph{Offset $\mathcal{C}$.} From the expanded form \eqref{eq:sp_C_expanded}:
\begin{equation}\label{eq:sp_C_case1}
\begin{split}
\mathcal{C} &= (1{+}G)(e^{-j\hat\beta}-1)\,e^{-j(\Theta-\hat\beta)} + (e^{-j\hat s}-1)\,e^{-j(\Theta-\hat\beta-\hat s)} \\
&\quad + (1{-}G)(e^{-j(\hat d-\hat s-\hat\beta)}-1)\,e^{-j(\Theta-\hat d)} + (-G)(e^{-j(\Theta-\hat d)}-1)
\end{split}
\end{equation}

Expanding each product $c_k(e^{-j\hat\alpha_k}-1)e^{-j\Phi_k} = c_k\,e^{-j(\hat\alpha_k+\Phi_k)} - c_k\,e^{-j\Phi_k} = c_k\,e^{-j(\Theta-\theta_{k-1})} - c_k\,e^{-j(\Theta-\theta_k)}$... 

More usefully, collecting all exponential terms:
\begin{equation}\label{eq:sp_C_case1_collected}
\begin{split}
\mathcal{C} &= (1{+}G)\,e^{-j\Theta} - (1{+}G)\,e^{-j(\Theta-\hat\beta)} + e^{-j(\Theta-\hat\beta)} - e^{-j(\Theta-\hat\beta-\hat s)} \\
&\quad + (1{-}G)\,e^{-j(\Theta-\hat\beta-\hat s)} - (1{-}G)\,e^{-j(\Theta-\hat d)} - G\,e^{-j(\Theta-\hat d)} + G
\end{split}
\end{equation}

Simplifying by grouping terms at common exponential arguments:
\begin{equation}\label{eq:sp_C_case1_simple}
\boxed{\mathcal{C}\big|_{\text{Case 1}} = (1{+}G)\,e^{-j\Theta} - G\,e^{-j(\Theta-\hat\beta)} + G\,e^{-j(\Theta-\hat\beta-\hat s)} - e^{-j(\Theta-\hat d)} + G}
\end{equation}

\paragraph{Initial state.}
\begin{equation}\label{eq:sp_z0_case1}
z_0 = \frac{\mathcal{C}}{1 + e^{-j\Theta}} = \frac{(1{+}G)\,e^{-j\Theta} - G\,e^{-j(\Theta-\hat\beta)} + G\,e^{-j(\Theta-\hat\beta-\hat s)} - e^{-j(\Theta-\hat d)} + G}{1 + e^{-j\Theta}}
\end{equation}

\paragraph{Coefficient $\mathcal{A}$.} With $r_1=-1$, $r_2=0$, $r_3=+1$, $r_4=+1$ and $\Theta_0=0$, $\Theta_1=\hat\beta$, $\Theta_2=\hat\beta+\hat s$, $\Theta_3=\hat d$:
\begin{equation}\label{eq:sp_A_case1}
\begin{split}
\mathcal{A} &= -(e^{-j\hat\beta}-1) + e^{-j(\hat\beta+\hat s)}(e^{-j(\hat d-\hat s-\hat\beta)}-1) + e^{-j\hat d}(e^{-j(\Theta-\hat d)}-1) \\
&= 1 - e^{-j\hat\beta} + e^{-j\hat d} - e^{-j(\hat\beta+\hat s)} + e^{-j\Theta} - e^{-j\hat d}
\end{split}
\end{equation}

The $e^{-j\hat d}$ terms cancel:
\begin{equation}\label{eq:sp_A_case1_simple}
\boxed{\mathcal{A}\big|_{\text{Case 1}} = 1 - e^{-j\hat\beta} - e^{-j(\hat\beta+\hat s)} + e^{-j\Theta}}
\end{equation}

\paragraph{Coefficient $\mathcal{B}$.} Using $\mathcal{B} = \sum r_k(-\mathcal{C}_{k-1} - c_k)(e^{-j\Delta\theta_k}-1)$ with $\mathcal{C}_0 = 0$:

The intermediate offsets are:
\begin{equation}
\begin{split}
\mathcal{C}_1 &= (1{+}G)(e^{-j\hat\beta}-1) \\
\mathcal{C}_2 &= (1{+}G)(e^{-j\hat\beta}-1)e^{-j\hat s} + (e^{-j\hat s}-1) \\
\mathcal{C}_3 &= \mathcal{C}_2\,e^{-j(\hat d-\hat s-\hat\beta)} + (1{-}G)(e^{-j(\hat d-\hat s-\hat\beta)}-1)
\end{split}
\end{equation}

The coefficient $\mathcal{B}$ has contributions from $k=1,3,4$ (where $r_k \ne 0$):
\begin{equation}\label{eq:sp_B_case1}
\begin{split}
\mathcal{B}\big|_{\text{Case 1}} &= -(- 0 - (1{+}G))(e^{-j\hat\beta}-1) \\
&\quad + (-\mathcal{C}_2 - (1{-}G))(e^{-j(\hat d-\hat s-\hat\beta)}-1) \\
&\quad + (-\mathcal{C}_3 - (-G))(e^{-j(\Theta-\hat d)}-1)
\end{split}
\end{equation}

Expanding $-\mathcal{C}_2 - (1{-}G) = -(1{+}G)e^{-j(\hat\beta+\hat s)} + Ge^{-j\hat s} + G$ and $-\mathcal{C}_3 + G = -(1{+}G)e^{-j\hat d} + Ge^{-j(\hat d-\hat\beta)} + Ge^{-j(\hat d-\hat s-\hat\beta)} + 1$, the factored form is:

\begin{equation}\label{eq:sp_B_case1_factored}
\begin{split}
\mathcal{B}\big|_{\text{Case 1}} &= (1{+}G)(e^{-j\hat\beta}-1) \\
&+ \left[-(1{+}G)e^{-j(\hat\beta+\hat s)} + Ge^{-j\hat s} + G\right](e^{-j(\hat d-\hat s-\hat\beta)}-1) \\
&+ \left[-(1{+}G)e^{-j\hat d} + Ge^{-j(\hat d-\hat\beta)} + Ge^{-j(\hat d-\hat s-\hat\beta)} + 1\right](e^{-j(\Theta-\hat d)}-1)
\end{split}
\end{equation}

Multiplying out all products and collecting terms at common exponential arguments (noting that intermediate exponentials cancel in pairs), the fully simplified form is:

\begin{equation}\label{eq:sp_B_case1_final}
\boxed{\begin{split}
\mathcal{B}\big|_{\text{Case 1}} &= (1{+}G)e^{-j\hat\beta} + (1{+}G)e^{-j(\hat\beta+\hat s)} - Ge^{-j\hat s} \\
&\quad - (1{+}G)e^{-j\Theta} + Ge^{-j(\Theta-\hat\beta)} + Ge^{-j(\Theta-\hat s-\hat\beta)} + e^{-j(\Theta-\hat d)} - 2(1{+}G)
\end{split}}
\end{equation}

\paragraph{Transconductance.} The exact $W$ for Case~1 is:
\begin{equation}\label{eq:sp_W_case1}
W\big|_{\text{Case 1}} = \frac{n}{Z_0\Theta}\,\mathrm{Re}\!\left(\frac{\mathcal{C}\cdot\mathcal{A}}{1+e^{-j\Theta}} + \mathcal{B}\right)
\end{equation}

with $\mathcal{C}$ from \eqref{eq:sp_C_case1_simple} and $\mathcal{A}$ from \eqref{eq:sp_A_case1_simple}.

\subsubsection{Case 2: $d < s + \beta$}

The intervals are $(c_k, \alpha_k) = \{(1{+}G,\,\beta),\; (1,\,d{-}\beta),\; (0,\,s{-}d{+}\beta),\; (-G,\,\pi{-}s{-}\beta)\}$ with rectification signs $r_k = \{-1,\, 0,\, 0,\, +1\}$.

\paragraph{Offset $\mathcal{C}$.} Since $c_3 = 0$, the third term vanishes:
\begin{equation}\label{eq:sp_C_case2}
\begin{split}
\mathcal{C} &= (1{+}G)(e^{-j\hat\beta}-1)\,e^{-j(\Theta-\hat\beta)} + (e^{-j(\hat d-\hat\beta)}-1)\,e^{-j(\Theta-\hat d)} \\
&\quad + (-G)(e^{-j(\Theta-\hat s-\hat\beta)}-1)
\end{split}
\end{equation}

Expanding and collecting:
\begin{equation}\label{eq:sp_C_case2_simple}
\boxed{\mathcal{C}\big|_{\text{Case 2}} = (1{+}G)\,e^{-j\Theta} - (1{+}G)\,e^{-j(\Theta-\hat\beta)} + e^{-j(\Theta-\hat\beta)} - e^{-j(\Theta-\hat d)} - G\,e^{-j(\Theta-\hat s-\hat\beta)} + G}
\end{equation}

Further simplification:
\begin{equation}
\mathcal{C}\big|_{\text{Case 2}} = (1{+}G)\,e^{-j\Theta} - G\,e^{-j(\Theta-\hat\beta)} - e^{-j(\Theta-\hat d)} - G\,e^{-j(\Theta-\hat s-\hat\beta)} + G
\end{equation}

\paragraph{Coefficient $\mathcal{A}$.} Only $r_1=-1$ and $r_4=+1$ contribute, with $\Theta_0=0$ and $\Theta_3 = \hat\beta + (\hat d-\hat\beta) + (\hat s-\hat d+\hat\beta) = \hat s + \hat\beta$:
\begin{equation}\label{eq:sp_A_case2}
\mathcal{A} = -(e^{-j\hat\beta}-1) + e^{-j(\hat s+\hat\beta)}(e^{-j(\Theta-\hat s-\hat\beta)}-1)
\end{equation}

Expanding:
\begin{equation}\label{eq:sp_A_case2_simple}
\boxed{\mathcal{A}\big|_{\text{Case 2}} = 1 - e^{-j\hat\beta} + e^{-j\Theta} - e^{-j(\hat s+\hat\beta)}}
\end{equation}

\paragraph{Observation.} Comparing \eqref{eq:sp_A_case1_simple} and \eqref{eq:sp_A_case2_simple}: the coefficient $\mathcal{A}$ has the \textit{same form} in both cases:
\begin{equation}\label{eq:sp_A_universal}
\mathcal{A} = 1 - e^{-j\hat\beta} + e^{-j\Theta} - e^{-j(\hat s+\hat\beta)}
\end{equation}

This is a consequence of the fact that $\mathcal{A}$ depends only on which intervals are rectified (and their cumulative start angles), not on the voltage levels $c_k$. The rectified intervals are always the first (negative-active) and the last (positive-active), with cumulative angles $0$ and $\hat s + \hat\beta$ respectively.

Using sum-to-product identities, $\mathcal{A}$ can be written in trigonometric form:
\begin{equation}\label{eq:sp_A_trig}
\mathcal{A} = -2j\sin\frac{\hat\beta}{2}\,e^{-j\hat\beta/2} - 2j\sin\frac{\Theta-\hat s-\hat\beta}{2}\,e^{-j(\hat s+\hat\beta+(\Theta-\hat s-\hat\beta)/2)}
\end{equation}

or more compactly:
\begin{equation}
\mathcal{A} = -2j\left[\sin\frac{\hat\beta}{2}\,e^{-j\hat\beta/2} + \sin\frac{\Theta-\hat s-\hat\beta}{2}\,e^{-j(\Theta+\hat s+\hat\beta)/2}\right]
\end{equation}

\paragraph{Coefficient $\mathcal{B}$.} Only $r_1=-1$ and $r_4=+1$ contribute. The intermediate offsets are:
\begin{equation}
\begin{split}
\mathcal{C}_1 &= (1{+}G)(e^{-j\hat\beta}-1) \\
\mathcal{C}_2 &= \mathcal{C}_1\,e^{-j(\hat d-\hat\beta)} + (e^{-j(\hat d-\hat\beta)}-1) \\
\mathcal{C}_3 &= \mathcal{C}_2\,e^{-j(\hat s+\hat\beta-\hat d)} \quad \text{(since $c_3=0$)}
\end{split}
\end{equation}

Expanding $-\mathcal{C}_3 + G = -(1{+}G)e^{-j(\hat\beta+\hat s)} + Ge^{-j\hat s} + e^{-j(\hat s+\hat\beta-\hat d)} + G$:

\begin{equation}\label{eq:sp_B_case2_factored}
\mathcal{B}\big|_{\text{Case 2}} = (1{+}G)(e^{-j\hat\beta}-1) + \left[-(1{+}G)e^{-j(\hat\beta+\hat s)} + Ge^{-j\hat s} + e^{-j(\hat s+\hat\beta-\hat d)} + G\right](e^{-j(\Theta-\hat s-\hat\beta)}-1)
\end{equation}

Multiplying out and collecting:

\begin{equation}\label{eq:sp_B_case2_final}
\boxed{\begin{split}
\mathcal{B}\big|_{\text{Case 2}} &= (1{+}G)e^{-j\hat\beta} + (1{+}G)e^{-j(\hat\beta+\hat s)} - Ge^{-j\hat s} - e^{-j(\hat s+\hat\beta-\hat d)} \\
&\quad - (1{+}G)e^{-j\Theta} + Ge^{-j(\Theta-\hat\beta)} + Ge^{-j(\Theta-\hat s-\hat\beta)} + e^{-j(\Theta-\hat d)} - (1{+}2G)
\end{split}}
\end{equation}

Note that at the case boundary $d = s + \beta$, the term $e^{-j(\hat s+\hat\beta-\hat d)} = e^0 = 1$, and the constant becomes $-(1+2G+1) = -2(1+G)$, matching Case~1 exactly. This confirms continuity of $\mathcal{B}$ across the switching pattern boundary.

\paragraph{Transconductance.}
\begin{equation}\label{eq:sp_W_case2}
W\big|_{\text{Case 2}} = \frac{n}{Z_0\Theta}\,\mathrm{Re}\!\left(\frac{\mathcal{C}\cdot\mathcal{A}}{1+e^{-j\Theta}} + \mathcal{B}\right)
\end{equation}

with $\mathcal{C}$ from \eqref{eq:sp_C_case2_simple} and $\mathcal{A}$ from \eqref{eq:sp_A_case2_simple}.

\subsubsection{Summary of explicit formulas}

The exact transconductance for arbitrary $(d, s, \beta, G, F_n)$ is given by:
\begin{equation}
W = \frac{n}{Z_0\Theta}\,\mathrm{Re}\!\left(\frac{\mathcal{C}\cdot\mathcal{A}}{1+e^{-j\Theta}} + \mathcal{B}\right)
\end{equation}

where $\mathcal{A}$ is universal \eqref{eq:sp_A_universal}, and $\mathcal{C}$ depends on the case:
\begin{itemize}
\item Case 1 ($d \ge s+\beta$): $\mathcal{C} = (1{+}G)e^{-j\Theta} - Ge^{-j(\Theta-\hat\beta)} + Ge^{-j(\Theta-\hat\beta-\hat s)} - e^{-j(\Theta-\hat d)} + G$
\item Case 2 ($d < s+\beta$): $\mathcal{C} = (1{+}G)e^{-j\Theta} - Ge^{-j(\Theta-\hat\beta)} - e^{-j(\Theta-\hat d)} - Ge^{-j(\Theta-\hat s-\hat\beta)} + G$
\end{itemize}

All quantities are elementary functions of the switching parameters and can be evaluated without any iterative computation.

\subsection{State-plane trajectories for Cases 1 and 2}

Figure~\ref{fig:sp_6panel} shows the exact state-plane trajectories for both switching pattern cases across buck ($G<1$), unity gain ($G=1$), and boost ($G>1$) operating points, all computed at normalized frequency $F_n = 1.3$. Each colored arc segment corresponds to one switching interval with its center marked on the real axis. The positive half-cycle is drawn with full opacity, the negative half-cycle (point-symmetric) with reduced opacity. The gray dashed ellipse represents the FHA trajectory for comparison.

\begin{figure}[h]
\centering
% this figure is generated by sim/state_plane_trajectories_6panel.py script
\includegraphics[width=\textwidth]{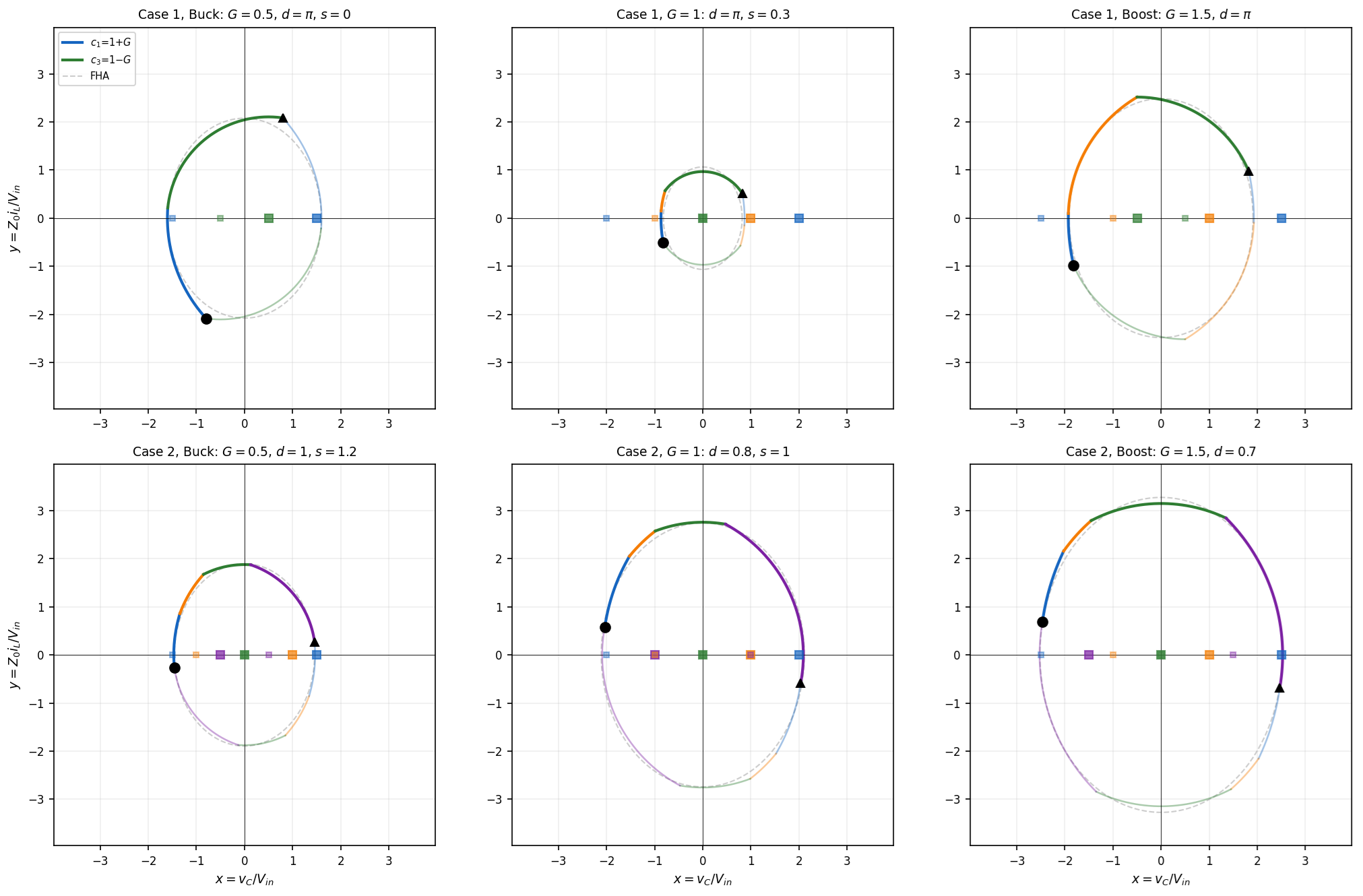}
\caption{State-plane trajectories at $F_n = 1.3$. Top row: Case~1 ($d \ge s+\beta$) with $d=\pi$. Bottom row: Case~2 ($d < s+\beta$). Columns: buck ($G=0.5$), unity ($G=1$), boost ($G=1.5$). Colors: $I_1$ (blue), $I_2$ (orange), $I_3$ (green), $I_4$ (purple). Gray dashed line: FHA ellipse.}
\label{fig:sp_6panel}
\end{figure}

\subsection{Exact zero-crossing angle} Within interval $k$, the state evolves as $z(\theta) = c_k + (z_{k-1} - c_k)e^{-j(\theta-\Theta_{k-1})}$, so:

\begin{equation}
\mathrm{Im}(z(\theta)) = \mathrm{Im}\!\left((z_{k-1}-c_k)e^{-j\tau}\right) = 0
\end{equation}

where $\tau = \theta - \Theta_{k-1} \in [0, \Delta\theta_k]$. Writing $w_k = z_{k-1} - c_k = |w_k|\,e^{j\phi_k}$:

\begin{equation}
|w_k|\sin(\phi_k - \tau) = 0 \quad \Longrightarrow \quad \tau = \phi_k + m\pi
\end{equation}

The derivative $\frac{d}{d\tau}\mathrm{Im}(z) = -|w_k|\cos(\phi_k - \tau)$. At $\tau = \phi_k$: derivative is $-|w_k| < 0$ (negative-going). At $\tau = \phi_k + \pi$: derivative is $+|w_k| > 0$ (positive-going). Therefore, the positive-going zero crossing is:

\begin{equation}\label{eq:sp_zero_crossing}
\tau_{\text{zero}} = \mathrm{arg}(z_{k-1} - c_k) + \pi
\end{equation}

provided this falls within $[0, \Delta\theta_k]$. The exact commutation angle is:
\begin{equation}\label{eq:sp_sigma_exact}
\boxed{\sigma_{\text{exact}} = F_n \cdot (\Theta_{k-1} + \tau_{\text{zero}}) = F_n\left(\Theta_{k-1} + \mathrm{arg}(z_{k-1} - c_k) + \pi\right)}
\end{equation}

where $k$ is the interval containing the crossing, and the intermediate states $z_{k-1}$ are computed by successive application of $z_k = c_k + (z_{k-1}-c_k)e^{-j\Delta\theta_k}$ starting from $z_0 = \mathcal{C}/(1+e^{-j\Theta})$. The offset $\mathcal{C}$ is the same quantity used for computing $W$, given by \eqref{eq:sp_C_case1_simple} for Case~1 or \eqref{eq:sp_C_case2_simple} for Case~2.

For typical operating conditions with the feedforward control law ($\sigma^* > 0$, $\delta^* \ge 0$), the zero crossing falls in interval 1 (where $c_1 = 1+G$), giving the simplest form:

\begin{equation}\label{eq:sp_sigma_interval1}
\sigma_{\text{exact}} = F_n \cdot \left(\mathrm{arg}\!\left(\frac{\mathcal{C}}{1+e^{-j\Theta}} - (1{+}G)\right) + \pi\right) \quad \text{(when crossing is in interval 1)}
\end{equation}

\subsection{Exact synchronous rectification condition}

In the original FHA-based analysis, the synchronous rectification condition $\delta = 0$ (secondary bridge switches exactly at the current zero crossing) yields the algebraic constraint $\cos\beta - G - \cos(\beta-d) - G\cos s = 0$. The state-plane framework provides the \textit{exact} version of this condition without harmonic truncation.

\subsubsection{General form}

The condition $\delta = 0$ means that $\mathrm{Im}(z) = 0$ (current crosses zero) exactly at $\theta = \hat\beta$ (the end of interval 1, where the secondary bridge transitions). Since $z_1 = c_1 + (z_0 - c_1)e^{-j\hat\beta}$:

\begin{equation}\label{eq:sp_sync_rect_general}
\boxed{\mathrm{Im}(z_1) = \mathrm{Im}\!\left((z_0 - c_1)\,e^{-j\hat\beta}\right) = 0}
\end{equation}

where $c_1 = 1+G$ and $z_0 = \mathcal{C}/(1+e^{-j\Theta})$. This condition is equivalent to:

\begin{equation}\label{eq:sp_sync_rect_arg}
\mathrm{arg}(z_0 - (1{+}G)) + \pi = \hat\beta
\end{equation}

which states that the first arc must subtend exactly angle $\hat\beta$ to bring the trajectory to the real axis.

Substituting $z_0$:

\begin{equation}\label{eq:sp_sync_rect_expanded}
\mathrm{Im}\!\left(\left(\frac{\mathcal{C}}{1+e^{-j\Theta}} - (1{+}G)\right) e^{-j\hat\beta}\right) = 0
\end{equation}

Writing $\frac{1}{1+e^{-j\Theta}} = \frac{e^{j\Theta/2}}{2\cos(\Theta/2)}$ and separating real and imaginary parts, this becomes a single transcendental equation relating $(d, s, \beta, G, F_n)$.

For both cases 1 and 2, the condition \eqref{eq:sp_sync_rect_expanded} takes the same form after substitution:

\begin{equation}\label{eq:sp_sync_rect_cases}
\mathrm{Im}\!\left(\left[\mathcal{C}(\hat\beta) - (1{+}G)(1+e^{-j\Theta})\right] e^{j(\Theta/2 - \hat\beta)}\right) = 0
\end{equation}

where $\hat\beta = \beta/F_n$, and $\mathcal{C}(\beta)$ is taken from \eqref{eq:sp_C_case1_simple} if $d \ge s + \beta$ (Case~1) or from \eqref{eq:sp_C_case2_simple} if $d < s + \beta$ (Case~2). In both cases, $\mathcal{C}$ depends on $\beta$ through the normalized angle $\hat\beta$ appearing in the exponentials, making \eqref{eq:sp_sync_rect_cases} an implicit transcendental equation in $\beta$. It is smooth and monotone, ensuring efficient numerical solution.

\subsubsection{Specialization: $d = \pi$, $s = 0$}

For the fully-driven converter without secondary shorting, the offset reduces to $\mathcal{C} = (1{+}G)e^{-j\Theta} - 2Ge^{-j(\Theta-\hat\beta)} - (1{-}G)$ (two intervals only). The numerator $\mathcal{N} = \mathcal{C} - (1{+}G)(1+e^{-j\Theta})$ simplifies to:

\begin{equation}
\mathcal{N} = -2G\,e^{-j(\Theta-\hat\beta)} - 2
\end{equation}

Applying the condition $\mathrm{Im}(\mathcal{N} \cdot e^{j(\Theta/2-\hat\beta)}) = 0$:

\begin{equation}
\mathrm{Im}\!\left(-2G\,e^{-j\Theta/2} - 2\,e^{j(\Theta/2-\hat\beta)}\right) = 0
\end{equation}

which gives $2G\sin(\Theta/2) - 2\sin(\Theta/2 - \hat\beta) = 0$, or:

\begin{equation}\label{eq:sp_sync_rect_d_pi_s0}
G\sin\frac{\Theta}{2} = \sin\!\left(\frac{\Theta}{2} - \hat\beta\right)
\end{equation}

Solving for $\hat\beta$:

\begin{equation}\label{eq:sp_sync_rect_d_pi_s0_solved}
\boxed{\hat\beta = \frac{\Theta}{2} - \arcsin\!\left(G\sin\frac{\Theta}{2}\right)}
\end{equation}

or equivalently $\beta = F_n\hat\beta = \frac{\pi}{2} - F_n\arcsin\!\left(G\sin\frac{\pi}{2F_n}\right)$. This is the \textit{exact closed-form} solution for the synchronous rectification phase shift, valid for all $F_n$ and $G \le 1/\sin(\Theta/2)$.

This is the exact synchronous rectification condition for the fully-driven converter with synchronous secondary ($d=\pi$, $s=0$), expressed as a single compact equation.

\subsubsection{Comparison with FHA condition and physical interpretation}

The FHA synchronous rectification condition for $d = \pi$, $s = 0$ is:
\begin{equation}\label{eq:sp_FHA_sync_rect}
\cos\beta = G \quad \Longleftrightarrow \quad \beta = \arccos G
\end{equation}

The exact state-plane condition \eqref{eq:sp_sync_rect_d_pi_s0} and the FHA condition \eqref{eq:sp_FHA_sync_rect} express the \textit{same physical requirement}: the tank current crosses zero precisely at the moment the secondary bridge transitions. The difference is:

\begin{itemize}
\item \textbf{FHA}: the \textit{fundamental component} of the tank current crosses zero at $\theta = \beta$. This is equivalent to requiring $\sigma_{FHA} = \mathrm{atan2}(B,A) = \beta$.

\item \textbf{State-plane}: the \textit{actual} tank current (including all harmonics) crosses zero at $\theta = \beta$. This is the physical condition $i_L(t_\beta) = 0$.
\end{itemize}

Since the actual current contains higher harmonics that shift the zero-crossing relative to the fundamental, the two conditions yield different values of $\beta$. For $G = 0.7$:

\begin{center}
\begin{tabular}{c|c|c}
$F_n$ & $\beta_{\text{exact}}$ & $\beta_{FHA} = \arccos G$ \\
\hline
1.05 & 0.760 & 0.795 \\
1.2 & 0.680 & 0.795 \\
1.5 & 0.594 & 0.795 \\
2.0 & 0.535 & 0.795 \\
5.0 & 0.481 & 0.795 \\
\end{tabular}
\end{center}

The discrepancy grows with $F_n$, reaching $18^\circ$ at $F_n = 5$. This is a consequence of the harmonic content: at higher switching frequencies, the ratio $\omega/\omega_0$ increases, but the arc structure becomes more pronounced (shorter arcs with larger angular steps between centers), causing the zero-crossing to shift further from the fundamental's prediction.

Conversely, near resonance ($F_n \to 1$), the exact condition approaches the FHA because the LC tank filters out harmonics more effectively and the trajectory becomes nearly elliptical.

The practical implication: if a controller uses $\beta = \arccos G$ (the FHA prescription) to achieve synchronous rectification, the actual zero-crossing will be offset from the switching edge by $\delta_{\text{actual}} = \beta - \sigma_{\text{exact}} \ne 0$. The exact condition \eqref{eq:sp_sync_rect_d_pi_s0} provides the corrected $\beta$ for true zero-crossing alignment.

\subsubsection{Exact inversion of commutation parameters}

The FHA-based feedforward control inverts $\sigma(d,s,\beta,G) = \sigma^*$, $\delta(d,s,\beta,G) = \delta^*$ to find switching parameters. The state-plane framework yields \textit{exact closed-form} inversions.

\paragraph{Problem formulation.} Given desired $\sigma^*$, $\delta^*$, voltage gain $G$, and $F_n$, find $d$, $s$, $\beta$ such that:
\begin{equation}\label{eq:sp_inversion_problem}
\sigma_{\text{exact}}(d, s, \beta, G, F_n) = \sigma^*, \quad \delta_{\text{exact}} = \beta - \sigma^* = \delta^*
\end{equation}

\paragraph{Step 1: $\beta$ is determined directly.} Since $\delta = \beta - \sigma$ by definition:
\begin{equation}\label{eq:sp_inv_beta}
\beta = \sigma^* + \delta^*
\end{equation}

\paragraph{Step 2: Closed-form solution for $d$ or $s$.}

The zero-crossing condition requires $\mathrm{Im}((z_0 - c_1)e^{-j\hat\sigma}) = 0$, i.e., the current crosses zero at resonant-time angle $\hat\sigma = \sigma^*/F_n$ on the first arc (center $c_1 = 1+G$). Since $\delta^* \ge 0$ implies $\hat\sigma \le \hat\beta$, the crossing always falls within interval~1.

Defining $\mathcal{N} = \mathcal{C} - (1{+}G)(1+e^{-j\Theta})$ and using $z_0 = \mathcal{C}/(1+e^{-j\Theta})$, the condition becomes:
\begin{equation}\label{eq:sp_inv_general_condition}
\mathrm{Im}\!\left(\mathcal{N}\,e^{j(\Theta/2-\hat\sigma)}\right) = 0
\end{equation}

\paragraph{Buck mode ($s = 0$, solve for $d$).} For $s = 0$, the numerator is $\mathcal{N} = -2Ge^{-j(\Theta-\hat\beta)} - e^{-j(\Theta-\hat d)} - 1$. Expanding \eqref{eq:sp_inv_general_condition} term by term and substituting $\hat\beta = \hat\sigma + \hat\delta$:

\begin{equation}
2G\sin(\Theta/2 - \hat\sigma - \Theta + \hat\beta) - \sin(\Theta/2 - \hat\sigma - \Theta + \hat d) - \sin(\Theta/2 - \hat\sigma) = 0
\end{equation}

Simplifying the arguments (using $\hat\beta = (\sigma^*+\delta^*)/F_n = \hat\sigma + \hat\delta$):
\begin{equation}
2G\sin(\hat\delta - \Theta/2) - \sin(\hat d - \hat\sigma - \Theta/2) - \sin(\Theta/2 - \hat\sigma) = 0
\end{equation}

Solving for $\hat d$:
\begin{equation}
\sin(\hat d - \hat\sigma - \Theta/2) = 2G\sin(\hat\delta - \Theta/2) - \sin(\Theta/2 - \hat\sigma)
\end{equation}

\begin{equation}\label{eq:sp_inv_d_closed}
\boxed{\hat d = \hat\sigma + \frac{\Theta}{2} + \arcsin\!\left(-2G\sin\!\left(\hat\delta - \frac{\Theta}{2}\right) - \sin\!\left(\frac{\Theta}{2} - \hat\sigma\right)\right)}
\end{equation}

where $\hat\sigma = \sigma^*/F_n$, $\hat\delta = \delta^*/F_n$, $\Theta = \pi/F_n$. A solution exists (buck mode valid) when the $\arcsin$ argument is in $[-1,1]$ and $d = F_n\hat d \le \pi$.

\paragraph{Boost mode ($d = \pi$, solve for $s$).} When the buck formula yields $d > \pi$, the boost mode applies. With $d = \pi$ (3 intervals: $(1{+}G,\beta)$, $(1,s)$, $(1{-}G,\pi{-}s{-}\beta)$), the numerator is $\mathcal{N} = -Ge^{-j(\Theta-\hat\beta)} - Ge^{-j(\Theta-\hat\beta-\hat s)} - 2$. Expanding \eqref{eq:sp_inv_general_condition} term by term with $\hat\beta = \hat\sigma+\hat\delta$:

\begin{equation}
-G\sin(\hat\delta - \Theta/2) - G\sin(\hat\delta + \hat s - \Theta/2) - 2\sin(\Theta/2 - \hat\sigma) = 0
\end{equation}

Using $-\sin(\hat\delta - \Theta/2) = \sin(\Theta/2-\hat\delta)$:

\begin{equation}
G\sin(\Theta/2-\hat\delta) - G\sin(\hat\delta + \hat s - \Theta/2) = 2\sin(\Theta/2 - \hat\sigma)
\end{equation}

Isolating the $\hat s$-dependent term:

\begin{equation}
\sin(\hat\delta + \hat s - \Theta/2) = \sin(\Theta/2-\hat\delta) - \frac{2}{G}\sin(\Theta/2-\hat\sigma)
\end{equation}

Solving for $\hat s$:
\begin{equation}\label{eq:sp_inv_s_closed}
\boxed{\hat s = \frac{\Theta}{2} - \hat\delta + \arcsin\!\left(\sin\!\left(\frac{\Theta}{2}-\hat\delta\right) - \frac{2}{G}\sin\!\left(\frac{\Theta}{2}-\hat\sigma\right)\right)}
\end{equation}

\paragraph{Verification.} Table~\ref{tab:inversion} compares the exact closed-form inversion with the FHA inversion across representative operating points. The column ``$\sigma$ actual'' shows the true $\sigma$ (from the exact state-plane model) when using the respective $d$ or $s$ value.

\begin{table}[h]
\centering
\caption{Inversion comparison: exact formula vs.\ FHA}
\label{tab:inversion}
\small
\begin{tabular}{cccc|cc|cc}
\toprule
$\sigma^*$ & $\delta^*$ & $G$ & $F_n$ & $d_{\text{exact}}$ & $\sigma$ actual & $d_{FHA}$ & $\sigma$ actual \\
\midrule
\multicolumn{8}{c}{\textit{Buck mode ($s=0$, solve for $d$)}} \\
\midrule
0.1 & 0.0 & 0.7 & 1.50 & 2.258 & 0.100 & 2.088 & 0.028 \\
0.2 & 0.0 & 0.7 & 1.50 & 2.422 & 0.200 & 2.204 & 0.110 \\
0.1 & 0.1 & 0.7 & 1.50 & 2.179 & 0.100 & 2.080 & 0.059 \\
0.3 & 0.0 & 0.5 & 1.30 & 2.009 & 0.300 & 1.916 & 0.251 \\
0.1 & 0.0 & 0.7 & 1.05 & 2.115 & 0.100 & 2.088 & 0.079 \\
0.2 & 0.1 & 0.6 & 1.20 & 2.037 & 0.200 & 1.986 & 0.173 \\
\midrule
& & & & $s_{\text{exact}}$ & $\sigma$ actual & $s_{FHA}$ & $\sigma$ actual \\
\midrule
\multicolumn{8}{c}{\textit{Boost mode ($d=\pi$, solve for $s$)}} \\
\midrule
0.1 & 0.0 & 1.3 & 1.50 & 0.934 & 0.100 & 1.011 & 0.156 \\
0.1 & 0.0 & 1.5 & 1.20 & 1.219 & 0.100 & 1.238 & 0.116 \\
0.2 & 0.0 & 1.3 & 1.30 & 1.032 & 0.200 & 1.038 & 0.205 \\
0.1 & 0.05 & 1.2 & 1.50 & 0.673 & 0.100 & 0.801 & 0.147 \\
\bottomrule
\end{tabular}
\end{table}

The exact formulas achieve the target $\sigma^*$ precisely in all cases, while the FHA inversion produces errors ranging from 2\% (near resonance, boost) to 72\% (higher $F_n$, buck). This is expected: the exact formulas explicitly depend on the normalized frequency $F_n$ through $\Theta = \pi/F_n$, whereas the FHA formulas are frequency-independent. Near resonance ($F_n \to 1$, $\Theta \to \pi$) the arc angles are large and the multi-arc trajectory closely resembles an ellipse, so both models agree. Far from resonance ($F_n \gg 1$, $\Theta \to 0$) the arcs become short and the harmonic distortion grows, making the FHA inversion increasingly inaccurate.

\paragraph{Comparison with FHA inversion.} The FHA inversion yields $d = \arccos(\cos\sigma^* - 2G\cos\delta^*) + \sigma^*$ (buck) and $s = \arccos(2\cos\sigma^*/G - \cos\delta^*) - \delta^*$ (boost). Placing the formulas side by side for the buck case:

\begin{equation}
\begin{array}{r@{\;=\;}l@{\qquad}l}
\hat d_{\text{exact}} & \hat\sigma + \dfrac{\Theta}{2} + \arcsin\!\left(-2G\sin\!\left(\hat\delta - \dfrac{\Theta}{2}\right) - \sin\!\left(\dfrac{\Theta}{2} - \hat\sigma\right)\right) & \text{(exact)} \\[10pt]
d_{FHA} & \sigma^* + \arccos\!\left(\cos\sigma^* - 2G\cos\delta^*\right) & \text{(FHA)}
\end{array}
\end{equation}

At resonance ($F_n = 1$, $\Theta = \pi$), all hatted quantities equal the switching angles ($\hat\sigma = \sigma^*$, $\hat\delta = \delta^*$, $\hat d = d$) and the exact formula becomes:
\begin{equation}
d = \sigma^* + \frac{\pi}{2} + \arcsin(2G\cos\delta^* - \cos\sigma^*)
\end{equation}

Using the identity $\arccos(x) = \pi/2 - \arcsin(x)$, the FHA formula rewrites as:
\begin{equation}
d = \sigma^* + \frac{\pi}{2} - \arcsin(\cos\sigma^* - 2G\cos\delta^*) = \sigma^* + \frac{\pi}{2} + \arcsin(2G\cos\delta^* - \cos\sigma^*)
\end{equation}

These are \textbf{identical}. The same equivalence holds for the boost formula. This proves that \textit{the FHA inversion is the exact state-plane inversion evaluated at resonance} ($F_n = 1$). The departure from FHA for $F_n > 1$ is entirely captured by the parameter $\Theta = \pi/F_n$, which compresses the trigonometric arguments as the switching frequency increases above resonance.

Figure~\ref{fig:q_comparison} illustrates this relationship using the aggregated switching parameter $q$, which combines both modes into a single continuous variable:

\begin{equation}\label{eq:sp_q_def}
q = \begin{cases} d & \text{buck mode } (d \le \pi) \\ s + \pi & \text{boost mode } (d = \pi) \end{cases}
\end{equation}

The FHA curve (black, frequency-independent) coincides with the $F_n = 1.0$ curve, confirming the resonance equivalence. As $F_n$ increases, the exact $q(G)$ curves depart progressively from the FHA, with the largest deviation near the buck-to-boost transition ($G \approx 1$).

\begin{figure}[h]
\centering
% this figure is generated by sim/plot_q_comparison.py script
\includegraphics[width=0.85\textwidth]{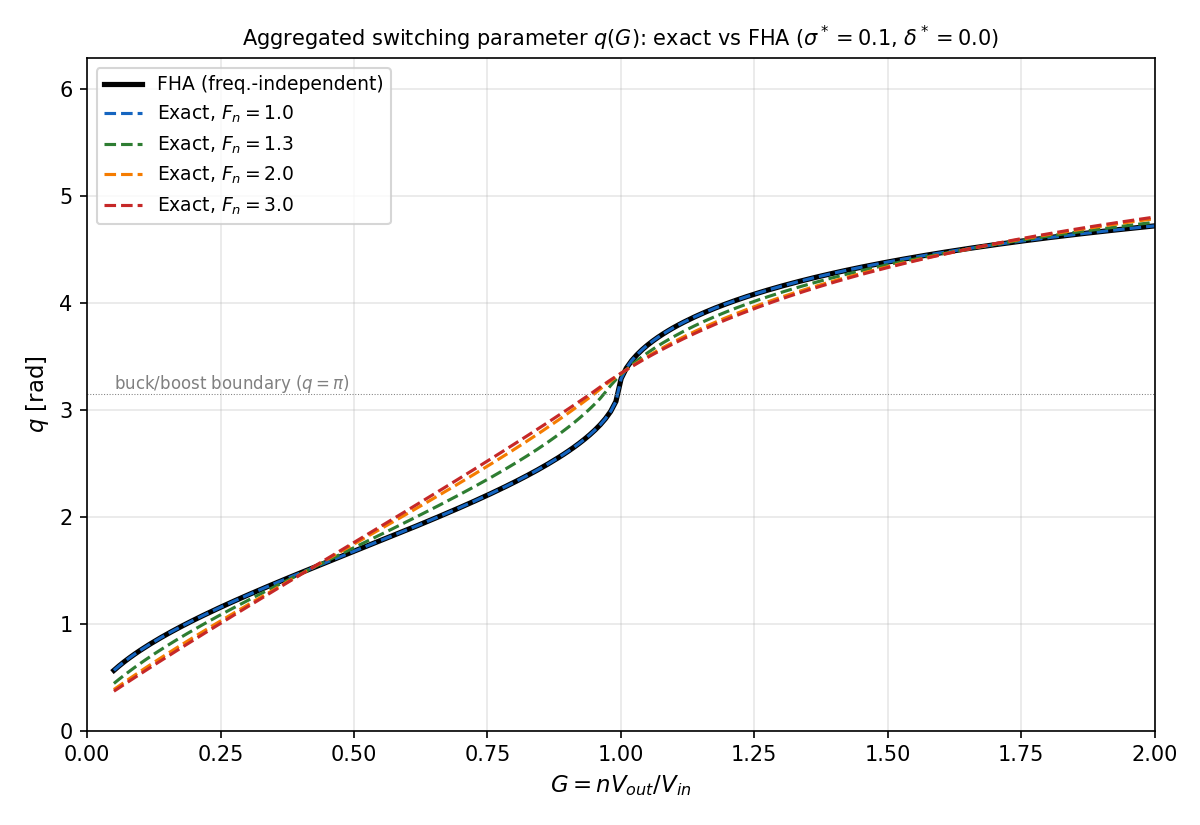}
\caption{Aggregated switching parameter $q(G)$ from exact state-plane inversion (dashed, for $F_n = 1.0, 1.3, 2.0, 3.0$) vs.\ FHA inversion (solid black, frequency-independent). Parameters: $\sigma^* = 0.1$, $\delta^* = 0$. At $F_n = 1$ (resonance) the exact and FHA curves coincide.}
\label{fig:q_comparison}
\end{figure}

\subsubsection{Low-power mode: inversion with $s_{add}$}

When the switching frequency reaches its maximum $\omega_{max}$ (i.e., $F_n = F_{n,max}$ is fixed), further output power reduction requires introducing additional secondary shorting $s_{add} > 0$. Following the original FHA-based approach, the aggregated variable $q$ with $s_{add}$ is defined as:

\begin{equation}\label{eq:sp_q_s_add}
\begin{pmatrix} d \\ s \end{pmatrix} = \begin{cases} \begin{pmatrix} q \\ s_{add} \end{pmatrix} & q \le \pi \quad \text{(buck)} \\ \begin{pmatrix} \pi \\ q - \pi \end{pmatrix} & q > \pi \quad \text{(boost)} \end{cases}
\end{equation}

In buck mode, $d = q$ and $s = s_{add}$; in boost mode, $d = \pi$ and $s = q - \pi$ (with $s_{add}$ absorbed into $s$).

\paragraph{Generalized buck inversion with $s_{add}$.} For the buck case with $s = s_{add} > 0$, the numerator $\mathcal{N} = -Ge^{-j(\Theta-\hat\beta)} - Ge^{-j(\Theta-\hat\beta-\hat s_{add})} - e^{-j(\Theta-\hat d)} - 1$ contains one $\hat d$-dependent term. Expanding $\mathrm{Im}(\mathcal{N}\,e^{j(\Theta/2-\hat\sigma)}) = 0$ and isolating the $\hat d$-dependent sine:

\begin{equation}
\sin(\hat d - \hat\sigma - \Theta/2) = -G\sin(\hat\delta - \Theta/2) - G\sin(\hat\delta + \hat s_{add} - \Theta/2) - \sin(\Theta/2 - \hat\sigma)
\end{equation}

Solving:

\begin{equation}\label{eq:sp_inv_d_s_add}
\boxed{\hat d = \hat\sigma + \frac{\Theta}{2} + \arcsin\!\left(-G\sin\!\left(\hat\delta - \frac{\Theta}{2}\right) - G\sin\!\left(\hat\delta + \hat s_{add} - \frac{\Theta}{2}\right) - \sin\!\left(\frac{\Theta}{2} - \hat\sigma\right)\right)}
\end{equation}

where $\hat s_{add} = s_{add}/F_n$. This generalizes the buck formula \eqref{eq:sp_inv_d_closed} (which is recovered for $s_{add} = 0$ since $-G\sin x - G\sin x = -2G\sin x$).

\paragraph{Low-power control algorithm.} The output transconductance at fixed $\omega = \omega_{max}$ is:

\begin{equation}
W(s_{add}) = \frac{n}{Z_0\Theta}\,\mathrm{Re}\!\left(\frac{\mathcal{C}(d(s_{add}),\, s_{add},\, \beta)\cdot\mathcal{A}}{1+e^{-j\Theta}} + \mathcal{B}\right)
\end{equation}

where $d(s_{add})$ is given by \eqref{eq:sp_inv_d_s_add} for the buck case (or $d = \pi$ with $s$ from \eqref{eq:sp_inv_s_closed} for boost). The algorithm for achieving a target $W^*$ below the minimum power at $s_{add} = 0$ is:

\begin{enumerate}
\item Compute $s_{add}^0$ (the minimal $s_{add}$ for monotonic $W$ decrease) from $W(s_{add}^0) = W(0)$.
\item For $s_{add} > s_{add}^0$, bisect on $s_{add}$ to satisfy $W(s_{add}) = W^*$.
\end{enumerate}

Each evaluation of $W(s_{add})$ uses the closed-form inversion \eqref{eq:sp_inv_d_s_add} to compute $d$, followed by the closed-form $W$ from \eqref{eq:sp_W_final_closed}. No iterative differential equation solving is required.

\paragraph{Control architecture.} The overall control system mirrors the parallel nonlinear compensation architecture from the original FHA-based design. The feedforward maps for commutation angles ($q$ and $\beta$ from \eqref{eq:sp_inv_d_s_add}, \eqref{eq:sp_inv_s_closed}, \eqref{eq:sp_inv_beta}) serve as the inner-loop nonlinear compensation, with external PI controllers adding corrections to $q$ and $\beta$ for robustness against model uncertainties. The outer-loop series compensation for output power uses the look-up table (or bisection) mapping $W^*$ to frequency $\omega$ and $s_{add}$.

The only structural difference from the FHA architecture is that the commutation feedforward maps \eqref{eq:sp_inv_d_s_add}--\eqref{eq:sp_inv_s_closed} operate on \textit{normalized resonant-time quantities} $\hat d = d/F_n$, $\hat s = s/F_n$, $\hat\beta = \beta/F_n$ rather than switching-frequency angles. Since $\hat x = x/F_n = x/(\omega/\omega_0)$, these hatted quantities have the physical dimension of \textit{time} (in units of $1/\omega_0$), not angles. This means the feedforward maps and the controllers could equivalently operate in time domain — controlling timing intervals between switching edges and current zero crossings rather than phase angles. The dependence on switching frequency enters naturally: all timing quantities scale with $1/F_n$, coupling the inner commutation loop to the outer frequency loop. This coupling reflects the physically correct dependency captured by the state-plane model.

\subsection{Recovery of First Harmonic Approximation}\label{sec:sp_FHA_recovery}

The FHA emerges as the leading-order approximation of the exact state-plane solution. In this section we provide a geometric interpretation of this relationship.

\subsubsection{Geometric view: single-arc approximation}

The exact state-plane trajectory consists of $N$ circular arcs, each with a different center $c_k$. The FHA is geometrically equivalent to replacing all $N$ arcs by a \textit{single arc} with an effective center $c_{\text{eff}}$.

For a single arc centered at $c_{\text{eff}}$ spanning the full half-period angle $\Theta$, the half-wave symmetry condition gives:
\begin{equation}
z_0^{\text{single}} = c_{\text{eff}} \cdot \frac{1-e^{-j\Theta}}{1+e^{-j\Theta}} = j\,c_{\text{eff}}\tan\frac{\Theta}{2}
\end{equation}

Comparing with the exact solution $z_0 = \mathcal{C}/(1+e^{-j\Theta})$, the effective center that matches the initial state is:
\begin{equation}\label{eq:sp_c_eff}
c_{\text{eff}} = \frac{\mathcal{C}}{1-e^{-j\Theta}}
\end{equation}

In general $c_{\text{eff}}$ is \textit{complex-valued}, meaning no single real-centered arc can perfectly reproduce the exact initial state. The imaginary part $\mathrm{Im}(c_{\text{eff}})$ measures the irreducible deviation of the multi-arc trajectory from any single-arc approximation.

\subsubsection{Curvature interpretation}

Each arc segment $k$ has radius $R_k = |z_{k-1} - c_k|$ and curvature $\kappa_k = 1/R_k$. The FHA trajectory (an ellipse) has continuously varying curvature. The key geometric difference:

\begin{itemize}
\item \textbf{Exact trajectory}: piecewise-constant curvature with discontinuous jumps at switching instants. Each jump in curvature is proportional to the voltage step $|c_{k+1} - c_k|$ divided by the local radius.

\item \textbf{FHA ellipse}: smooth curvature varying between $F_n/Y_1$ (at current peaks) and $F_n^2/Y_1$ (at voltage peaks), where $Y_1$ is the fundamental current amplitude.
\end{itemize}

The FHA becomes exact when all arc centers coincide ($c_1 = c_2 = \ldots = c_N$), which corresponds to constant applied voltage (no switching) — a degenerate case. The FHA error is fundamentally related to the \textit{spread of arc centers} relative to the trajectory size.

\subsubsection{Fourier interpretation}

The applied voltage $V_n(\theta)$ has Fourier harmonics $V_k$ at frequencies $k\omega$. Each harmonic excites a current component at the same frequency with amplitude $Y_k = V_k/|Z(k\omega)|$. In the state plane:

\begin{itemize}
\item The fundamental ($k=1$) produces the FHA ellipse with semi-axes $Y_1/F_n$ (voltage) and $Y_1$ (current).
\item The $k$-th harmonic produces a smaller ellipse with semi-axes $Y_k/(kF_n)$ and $Y_k$, rotating $k$ times faster.
\item The exact trajectory is the \textit{superposition} of all harmonic ellipses.
\end{itemize}

The multi-arc piecewise-circular structure is therefore the result of adding infinitely many nested ellipses of decreasing size — the FHA retains only the largest one.

\subsubsection{Error bound}

The relative discrepancy between $W_{\text{exact}}$ and $W_{FHA}$ is bounded by the higher-harmonic current contributions:
\begin{equation}\label{eq:sp_error_bound}
\frac{|W_{\text{exact}} - W_{FHA}|}{W_{\text{exact}}} \lesssim \sum_{k=3,5,\ldots} \frac{|V_k|}{|V_1|} \cdot \frac{|F_n - 1/F_n|}{|kF_n - 1/(kF_n)|}
\end{equation}

The first factor is the voltage harmonic ratio ($\le 1/k$ for half-wave symmetric waveforms). The second is the impedance attenuation. For above-resonance operation, the $k$-th harmonic current is suppressed by the LC tank impedance as $\sim 1/(k \cdot kF_n) = 1/k^2 F_n$, giving rapid convergence of the series.

Note that while the \textit{pointwise} trajectory deviation between exact and FHA can be significant (several percent of the current amplitude at switching instants), the \textit{integrated} quantities ($W$, $\sigma$) benefit from cancellation of positive and negative deviations, resulting in much smaller errors in practice ($<2\%$ for typical operating conditions).

\subsection{Summary of exact analytical results}

\begin{theorem}[Exact Steady-State of DB SRC]\label{thm:exact_ss}
For a dual-bridge series resonant converter with switching parameters $(d, s, \beta)$, normalized frequency $F_n = \omega/\omega_0 \ne \text{integer}$, and voltage gain $G$:

\textbf{1. Switching pattern:} Determine intervals $(c_k, \alpha_k)_{k=1}^N$ from \eqref{eq:sp_case1} or \eqref{eq:sp_case2}. Set $\Delta\theta_k = \alpha_k/F_n$.

\textbf{2. Offset computation:} Calculate $\mathcal{C}$ by the recurrence:
\begin{equation}
\mathcal{C}_0 = 0, \quad \mathcal{C}_k = (\mathcal{C}_{k-1} + c_k)e^{-j\Delta\theta_k} - c_k
\end{equation}

\textbf{3. Initial state:}
\begin{equation}
z_0 = x_0 + jy_0 = \frac{\mathcal{C}}{1 + e^{-j\Theta}}, \quad \Theta = \pi/F_n
\end{equation}

\textbf{4. Exact transconductance:}
\begin{equation}
W = \frac{n}{Z_0\Theta}\,\mathrm{Re}\!\left(\frac{\mathcal{C}\,\mathcal{A}}{1+e^{-j\Theta}} + \mathcal{B}\right)
\end{equation}

where $\mathcal{A}$, $\mathcal{B}$ are defined in \eqref{eq:sp_AB_coeff}.

\textbf{5. Exact zero crossing:}
\begin{equation}
\sigma_{\text{exact}} = F_n\left(\Theta_{k-1} + \mathrm{arg}(z_{k-1}-c_k) + \pi\right)
\end{equation}

in the interval $k$ containing the positive-going zero crossing.

\textbf{6. FHA as special case:} The first harmonic approximation \eqref{eq:FHA_AB}--\eqref{eq:FHA_model} corresponds to replacing the multi-arc trajectory by a single effective arc, with error bounded by \eqref{eq:sp_error_bound}.
\end{theorem}

\begin{proof}[Proof sketch]
\textit{(1)} follows from the definition of bridge switching waveforms in Section~2 and the enumeration of voltage levels $v_{in} - nv_{out}$ across all combinations of bridge states.

\textit{(2)} The recurrence implements the composition $z \mapsto c_k + (z-c_k)e^{-j\Delta\theta_k}$ applied to $z=0$, computing the negated affine offset of the composed map. This is verified by noting that propagating any $z_0$ gives $z_N = e^{-j\Theta}z_0 - \mathcal{C}$ (the total rotation plus accumulated shift).

\textit{(3)} Half-wave symmetry requires $z_N = -z_0$. Substituting $z_N = e^{-j\Theta}z_0 - \mathcal{C}$ gives $(1+e^{-j\Theta})z_0 = \mathcal{C}$. The denominator $1+e^{-j\Theta} = 2\cos(\Theta/2)e^{-j\Theta/2} \ne 0$ for $\Theta \ne \pi$ (i.e., $F_n \ne 1$), guaranteeing existence and uniqueness.

\textit{(4)} The output current integral $\int r(\theta)\,y\,d\theta = \sum r_k(x_k-x_{k-1}) = \sum r_k\,\mathrm{Re}(\Delta z_k)$ is linear in $z_0$ since each $\Delta z_k = (z_{k-1}-c_k)(e^{-j\Delta\theta_k}-1)$ and $z_{k-1} = e^{-j\Theta_{k-1}}z_0 - \mathcal{C}_{k-1}$. Collecting terms linear in $z_0$ gives coefficient $\mathcal{A}$; the remainder is $\mathcal{B}$.

\textit{(5)} Within arc $k$, $\mathrm{Im}(z(\tau)) = |w_k|\sin(\phi_k - \tau)$ where $w_k = z_{k-1}-c_k = |w_k|e^{j\phi_k}$. The positive-going zero ($d[\mathrm{Im}]/d\tau > 0$) occurs at $\tau = \phi_k + \pi = \mathrm{arg}(w_k)+\pi$.

\textit{(6)} When all $\Delta\theta_k \to 0$, the multi-arc trajectory approaches a smooth curve whose shape is determined by the fundamental Fourier component of the applied voltage, as shown in Section~\ref{sec:sp_FHA_recovery}.
\end{proof}

%\begin{figure}[t]
%\input{Freq_Ctrl_parallel.TpX}  % TODO: adapt figure from original paper
%\centering
%\caption{Controller architecture with commutation parameters control $\sigma$ and $\delta$, and output power control $W$.}
%\label{fig:Freq_Ctrl_parallel}
%\end{figure}

\end{document}